\documentclass[11pt,a4paper,reqno, oneside]{amsart}

\usepackage{amsopn}
\usepackage{amsmath}
\usepackage{amssymb}
\usepackage{amsthm}
\usepackage{amsaddr}
\usepackage{hyperref}
\usepackage{graphicx}
\usepackage{color}
\usepackage{listings}
\usepackage{mathtools}
\usepackage{subcaption}
\usepackage{enumerate}
\usepackage{cite}
\usepackage{todonotes}
\usepackage{amsfonts}

\usepackage[ruled,vlined]{algorithm2e}

\usepackage{multirow}
\usepackage{adjustbox}
\usepackage{geometry}
\usepackage{array}

\usepackage{tikz}
\usetikzlibrary{angles,quotes}
\usetikzlibrary{decorations.pathreplacing}

\newcommand{\ie}{{\emph{i.e.\/}}}

\DeclarePairedDelimiter\ceil{\lceil}{\rceil}
\DeclareMathOperator{\tr}{tr}

\DeclareMathOperator{\diag}{diag}

\newcommand{\R}{\ensuremath{\mathbb{R}}}
\newcommand{\N}{\ensuremath{\mathbb{N}}}

\newcommand{\C}{\ensuremath{\mathbb{C}}}

\newcommand{\ket}[1]{\ensuremath{|#1\rangle}}
\newcommand{\bra}[1]{\ensuremath{\langle#1|}}
\newcommand{\ketbra}[2]{\ensuremath{\ket{#1} \! \bra{#2}}}
\newcommand{\proj}[1]{\ensuremath{\ketbra{#1}{#1}}}
\newcommand{\braket}[2]{\ensuremath{\langle{#1}|{#2}\rangle}}

\newcommand{\1}{{\rm 1\hspace{-0.9mm}l}}
\newcommand{\Id}{{\rm 1\hspace{-0.9mm}l}}

\newcommand{\ee}{\ensuremath{\mathrm{e}}}

\newcommand{\ii}{\ensuremath{\mathrm{i}}}

\newcommand{\DD}{\mathcal{D}}

\newcommand{\PP}{\mathcal{P}}

\newcommand{\UU}{\mathcal{U}}

\newcommand{\DU}{\mathcal{DU}}

\newcommand{\diaguni}{\ensuremath{\mathcal{DU}}}

\newtheorem{lemma}{Lemma}
\newtheorem{theorem}{Theorem}
\newtheorem{corollary}{Corollary}

\newtheorem{remark}{Remark}
\newtheorem{example}{Example}

\title{On the optimal certification of von Neumann measurements}

\author{Paulina Lewandowska$^{1}$, Aleksandra Krawiec$^{1}$, Ryszard
Kukulski$^1$, \L ukasz Pawela* $^{1}$, \and Zbigniew Pucha\l a$^{1,2}$}

\address{$^1$Institute of Theoretical and Applied Informatics, Polish
Academy of Sciences, ul. Ba{\l}tycka 5, 44-100 Gliwice, Poland}
\address{$^2$ Faculty of Physics, Astronomy and Applied Computer Science,
Jagiellonian University, ul. {\L}ojasiewicza 11, 30-348 Krak{\'o}w, Poland}

\begin{document}
\maketitle
$^{*}$ \normalsize{E-mail address: \url{lpawela@iitis.pl}}
\begin{abstract}
In this report we study certification of quantum measurements, which can be
viewed as the extension of quantum hypotheses testing. This extension involves
also the study of the input state and the measurement procedure. Here, we will
be interested in two-point (binary) certification scheme in which the null and
alternative hypotheses are single element sets. Our goal is to minimize the
probability of the type II error given some fixed statistical significance. In
this report, we begin with studying the two-point certification of pure quantum
states and unitary channels to later use them to prove our main result, which is
the certification of von Neumann measurements in single-shot and parallel
scenarios. From our main result follow the conditions when two pure states,
unitary operations and von Neumann measurements cannot be distinguished
perfectly but still can be certified with a given statistical significance.
Moreover, we show the connection between the certification of quantum channels
or von Neumann measurements and the notion of $q$-numerical range.
\end{abstract}

\section{Introduction}

The validation of sources producing quantum states and measurement devices,
which are involved in quantum computation workflows, is a necessary step of
quantum technology
\cite{aolita2015reliable,spagnolo2014experimental,chabaud2020efficient}. The
search for practical and reliable tools for validation of quantum architecture
has attracted a lot of attention in recent years
\cite{carolan2014experimental,chareton2020toward,
wu2019efficient,jiang2020towards,tiedau2020benchmarking}. Rapid technology
development and increasing interest in quantum computers paved the way towards
creating more and more efficient validation methods of Noisy Intermediate-Scale
Quantum devices (NISQ) \cite{preskill2018quantum, gambetta2019benchmarking}.
Such a growth comes along with ever-increasing requirements for the precision of
the components of quantum devices. The tasks of ensuring the correctness of
quantum devices are referred to as \emph{validation}.

Let us begin with sketching the problem of validation of quantum architectures.
Imagine you are given a black box and are promised two things. First, it
contains a pure quantum state (or a unitary matrix or a von Neumann POVM), and
second, it contains one of two possible choices of these objects. The owner of
the box, Eve, tells you which of the two possibilities is contained within the
box. Yet, for some reason, you do not completely trust her and decide to perform
some kind of hypothesis testing scheme on the black box. You decide to take
Eve's promise as the null hypothesis, $H_0$, for this scheme and the second of
the possibilities as the alternative hypothesis, $H_1$. Since now you own the
box and are free to proceed as you want, you need to prepare some input into the
box and perform a measurement on the output. A particular input state and final
measurement (or only the measurement, for the case when the box contains a
quantum state) will be called a certification strategy. Of course, just like in
classical hypothesis testing, in our certification scheme we have two possible
types of errors. The type I error happens if we reject the null hypothesis when
it was in reality true. The type II error happens if we accept the null
hypothesis when we should have rejected it. The main aim of certification is
finding the optimal strategy which minimizes  one type of error when the other
one is fixed.  In this work we are interested in the minimization of the type II
error given a fixed type I error. This approach will be called
\emph{certification}.

Certification of quantum objects is closely related with the other well-know
method of the validation, that is the problem of discrimination of those
objects. Intuitively, in the discrimination problem we are given one of two
quantum objects sampled according to a given a priori probability distribution.
Hence, the probability of making an error in the discrimination task is equal to
the average of the type I and type II errors over the assumed probability
distribution. Therefore, the discrimination problem can be seen as
\emph{symmetric distinguishability}, as opposed to certification, that is
\emph{asymmetric distinguishability}. In other words, the main difference
between both approaches is that the main task of discrimination is the
minimization over the average of both types of possible errors while the
certification concerns the minimization over one type of error when the bound of
the other one is assumed.

While in the basic version of the aforementioned scenario we focus on the case
when the validated quantum  object can be used exactly once, one can consider
also the situation in which this object can be utilized multiple times in
various configurations. In the parallel scheme the validated object can be used
many times, but no processing can be performed between the usages of this
object. In the adaptive scenario however, we are allowed to perform any
processing we want between the uses of the validated quantum object.

The problem of discrimination of quantum states and channels was solved
analytically by Helstrom  a few decades ago in~\cite{helstrom1976quantum,
watrous2018theory}. The multiple-shot scenario of the discrimination of quantum
states was studied in \cite{namkung2018analysis,namkung2018sequential} whereas
the discrimination of quantum channels in the multiple-shot scenario was
investigated for example in
\cite{duan2009perfect,duan2016parallel,duan2007entanglement,
bae2015discrimination,cao2016minimal}. Examples of channels which cannot be
discriminated perfectly in the parallel scheme, but nonetheless can be
discriminated perfectly using the adaptive approach, were discussed in
\cite{harrow2010adaptive,krawiec2020discrimination}. The work
\cite{ji2006identification} paved the way for studying the discrimination of
quantum measurements. Therefore, this work can be seen as a natural extension of
our works \cite{puchala2018strategies} and \cite{puchala2018multiple}, where we
studied the discrimination of von Neumann measurements in single and
multiple-shot scenarios, respectively.  Nevertheless, one can also consider a
scenario in which we are allowed to obtain an inconclusive answer. Therefore, we
arrive at the unambiguous discrimination of quantum operations discussed in
\cite{puchala2018multiple,wang2006unambiguous}.

The problem of certification of quantum states and channels has not been studied
as exhaustively as their discrimination. The certification of quantum objects
was first studied by Helstrom in \cite{helstrom1976quantum}, where the problem
of pure state certification was considered. Further, certification schemes were
established to the case of mixed states in~\cite{audenaert2012quantum}. A
natural extension of quantum state certification is the certification of
unitary operations. This problem was solved in~\cite{lu2010optimal}.
Considerations about multiple-shot scenario of certification of quantum states 
and unitary channels were investigated in   \cite{lu2010optimal}.

All the above-mentioned approaches towards the certification were considered in
a finite number of steps. Another common approach  involves studying
certification of quantum objects in the asymptotic regime
\cite{mosonyi2015quantum, ogawa2004error, ogawa2005strong}
which assumes that the number of copies of the given quantum object goes to
infinity. It focuses on  studying the convergence of the probability of making
one type of error while a bound on the second one is assumed. This task is
strictly related with the term of relative entropy and its asymptotic
behavior~\cite{nagaoka2006converse}. For a more general overview of quantum
certification we refer the reader to \cite{helstrom1969quantum,
eisert2020quantum}.

This work will begin with recalling one-shot certification scenario which will
later be extended to the multiple-shot case. For this purpose, we will often
make use of the terms of numerical range and $q$-numerical range as essential
tools in the proofs \cite{li1998q,li1998some,tsing1984constrained,nr}.  More
specifically, one of our results presented in this work is a geometric 
interpretation of the formula for minimized probability of the type II error in 
the problem of certification of unitary channels, which is strictly connected 
with the notion of $q$-numerical range. 
Later, basing on the results on the certification of unitary channels we
will extend these considerations to the problem of certification of von Neumann
measurements. It will turn out that the formula for minimized probability of the
type II error can also be connected with the notion of $q$-numerical range. On
top of that, we will show that entanglement can significantly improve the
certification of von Neumann measurements. Eventually, we will prove that the
parallel certification scheme is optimal.

This work is organized as follows. We begin with preliminaries in
Section~\ref{sec:preliminaries}. Then, in
Section~\ref{sec:two_point_certification_of_states} we present the two-point
certification of pure quantum states. Certification of unitary channels is
discussed in Section~\ref{sec:unitary_channels_two_points}. After presenting the
known results we introduce geometrical interpretation of the problem of
certification of unitary channels, expressed in terms of $q$-numerical range.
The certification of von Neumann measurements is studied in
Section~\ref{sec:von_neumann_measurements} and our main result is stated therein
as Theorem~\ref{thm_measurements}. Section~\ref{sec:multiple} generalizes the
results on certification to the multiple-shot scenario and the optimality of the
parallel scheme for certification of von Neumann measurements is presented as
Theorem~\ref{th:parallel_optimal}.

\section{Preliminaries}\label{sec:preliminaries}

Let $M_{d_1,d_2}$ be the set of all matrices of dimension $d_1 \times d_2$ over
the field $\C$. For the sake of simplicity, square matrices will be denoted by
$M_d$. The set of quantum states, that is positive semidefinite operators having
trace equal to one, will be denoted $\DD_d$. By default, when we write
$\ket{\psi}, \ket{\varphi}$, we mean normalized pure states, unless we mention
otherwise. The subset of $M_d$ consisting of unitary matrices will be denoted by
$\UU_d$, while its subgroup of diagonal unitary operators will be denoted by
$\DD \UU_d$.  Let $U \in \UU_d$ be a unitary matrix. A unitary channel
$\Phi_{U}$ is defined as $\Phi_U(\cdot) = U \cdot U^\dagger$. A general quantum
measurement, that is a positive operator valued measure (POVM) $\PP$ is a
collection of positive semidefinite operators $\{E_1, \ldots, E_m \}$ called
\emph{effects}, which sum up to identity, \ie $ \, \, \sum_{i=1}^m E_i = \1$. If
all the effects are rank-one projection operators, then such a measurement is
called von Neumann measurement. Every von Neumann measurement can be
parameterized by a unitary matrix and hence we will use the notation $\PP_{U}$
for a von Neumann measurement with effects $\{\proj{u_1}, \ldots, \proj{u_d}\}$,
where $\ket{u_i}$ is the $i$-th column of the unitary matrix $U$. The action of
quantum measurement $\PP_{U}$ on some state $\rho \in \mathcal{D}_d$ can be
expressed as the action of a quantum channel

\begin{equation}
\PP_{U} : \rho \rightarrow \sum_{i=1}^d \bra{u_i} \rho \ket{u_i} \proj{i}.
\end{equation}

As mentioned in the Introduction, in this work we focus on two-point hypothesis
testing of quantum objects.
The starting point towards the certification of quantum objects is the
hypothesis testing of quantum states. Let $H_0$ be a null hypothesis which
states that the obtained state was $\ket{\psi}$, while the alternative
hypothesis, $H_1$, states that the obtained state was $\ket{\varphi}$. The
certification is performed by the use of a binary measurement $\{ \Omega, \1-
\Omega \}$, where the effect $\Omega$ corresponds to accepting the null
hypothesis and  $\1-\Omega$ accepts the alternative  hypothesis. In this work we
will be considering only POVMs with two effects of this form. Therefore the
effect $\Omega$ uniquely determines the POVM and hence we will be using the
words measurement and effect interchangeably.

Assume we have a fixed measurement $\Omega$. We introduce the probability of the
type I error, $p_\text{I}(\Omega)$, that is the probability of rejecting the
null hypothesis when in fact it was true, as
\begin{equation}
p_\text{I}(\Omega) = \tr \left((\1-\Omega) \proj{\psi} \right) = 1 - \tr \left(
\Omega \proj{\psi} \right).
\end{equation}
The type II error, $p_{\text{II}}(\Omega)$, that is the probability of accepting
the null hypothesis $H_0$ when in reality $H_1$ occurred, is defined as
\begin{equation}
p_\text{II}(\Omega) = \tr \left(\Omega \proj{\varphi} \right).
\end{equation}
In the remainder of this work we will assume the statistical significance
$\delta \in [0,1]$, that is the probability of the type I error will be
upper-bounded by $\delta$. Our goal will be to find a most powerful test, that
is to minimize the probability of the type II error by finding the optimal
measurement, which we will denote as $\Omega_0$. Such $\Omega_0$, which
minimizes $p_\text{II}(\Omega)$ while assuming the  statistical significance
$\delta$, will be called an \emph{optimal measurement}. The minimized
probability of type II error will be denoted by
\begin{equation}
p_{\text{II}} \coloneqq \min_{\Omega: p_{\text{I}}(\Omega) \leq \delta}
p_{\text{II}}(\Omega).
\end{equation}

While certifying quantum channels and von Neumann measurements, we will also
need to minimize over input states. Let a channel $\Phi_0$ correspond to
hypothesis $H_0$ and $\Phi_1$ correspond to hypothesis $H_1$. We define
\begin{equation}
\begin{split}
p_{\text{I}}^{\ket{\psi}}(\Omega) &= \tr \left((\1-\Omega)\Phi_0(\proj{\psi})
\right)\\
p_{\text{II}}^{\ket{\psi}}(\Omega) &= \tr \left(\Omega \Phi_1(\proj{\psi})
\right).
\end{split}
\end{equation}
Naturally, for each input state we can consider minimized probability of type
II error, that is
\begin{equation}
p_{\text{II}}^{\ket{\psi}}=\min_{\Omega: p_{\text{I}}^{\ket{\psi}}(\Omega) \leq
\delta} p_{\text{II}}^{\ket{\psi}}(\Omega).
\end{equation}
Finally, we will be interested in calculating optimized probability of type II
error over all input states. This will be denoted as
\begin{equation}
p_{\text{II}} \coloneqq \min_{\ket{\psi}} p_{\text{II}}^{\ket{\psi}}.
\end{equation}
Note that the symbol $p_{\text{II}}$ is used in two contexts. In the problem of
certification of states the minimization is performed only over measurements
$\Omega$, while in the problem of certification of unitary channels and von
Neumann measurements the minimization is over both measurements $\Omega$ and
input states $\ket{\psi}$. In other words, $p_{\text{II}}$ is equal to the
optimized probability of the type II error in certain certification problem.

The input state which minimizes $p_{\text{II}}$ will be called an \emph{optimal
state}. We will use the term \emph{optimal strategy} to denote both the optimal
state and the optimal measurement.

Now, we introduce a basic toolbox for studying the certification of quantum
objects which is strictly related with the problem of discrimination of quantum
channels. First, we will be using the notion of the diamond norm. The diamond
norm of a superoperator $\Psi$ is defined as
\begin{equation}
\| \Psi \|_\diamond
\coloneqq \max_{\|X\|_1 = 1} \| \left(\Psi \otimes \1\right) (X) \|_1.
\end{equation}
The celebrated theorem of Helstrom \cite{helstrom1976quantum} gives a lower
bound on the probability of making an error in distinction in the scenario of
symmetric discrimination of quantum channels. The probability of incorrect
symmetric discrimination between quantum channels $\Phi$ and $\Psi$ is bounded
as follows
\begin{equation}
p_{e} \ge  \frac12 - \frac14 \| \Phi - \Psi \|_\diamond.
\end{equation}

Moreover, our results will often make use of the terms of numerical range and
$q$-numerical range~\cite{nr}. The numerical range is a subset of complex plane
defined for a matrix $X \in M_d$ as
\begin{equation}
W(X) := \{  \bra{\psi} X \ket{\psi}: \braket{\psi}{\psi}=1 \}
\end{equation}
while the $q$-numerical range
\cite{li1998q,tsing1984constrained,li1998some}
is defined for a matrix $X\in M_d$ as
\begin{equation}
W_q(X) := \{  \bra{\psi} X \ket{\varphi}: \braket{\psi}{\psi} = \braket{\varphi}{\varphi} = 1, \, \braket{\psi}{\varphi} =q,\, q \in \C \}.
\end{equation}
The standard numerical range is the special case of $q$-numerical range for
$q=1$, that is $W(X) = W_1(X)$.
We will use the notation
\begin{equation}\label{eq:dist_q_nr_to_zero_def}
\nu_{q}(X):= \min\{ |x|: x \in W_q(X) \}
\end{equation}
to denote the distance on a complex plane from $q$-numerical range to zero. In
the case  when $q=1$, we will simply write $\nu(X)$. The main properties of
$q$-numerical range are its convexity and
compactness~\cite{tsing1984constrained}. The detailed shape of $q$-numerical
range is described in \cite{li1998some}. The properties of $q$-numerical
range~\cite{duan2009perfect} that will be used throughout this paper are
\begin{equation}\label{properties-q-nr-inclusions}
W_{q'} \subseteq \frac{q'}{q} W_{q} \quad \text{for} \quad q \leq q', \quad q,q' \in \R
\end{equation} and
\begin{equation}\label{properties-q-nr-tensor-product}
W_q (X \otimes \1) = W_q(X), \quad q\in \R.
\end{equation}
From the above it is easy to see that \begin{equation}
\nu_q (X \otimes \1) = \nu_q(X), \quad q\in \R.
\end{equation}
In the  we provide an animation of $q$-numerical range
of unitary matrix $U \in \UU_3$ with eigenvalues $1, \ee^{ \frac{\pi \ii}{3}}$
and $\ee^{ \frac{2\pi \ii}{3}}$ for all parameters $q \in [0,1]$.

\section{Two-point certification of pure
states}\label{sec:two_point_certification_of_states}

In this section we recall the results concerning the certification of pure
quantum states. We state the optimized  probability of the type II error for the
quantum hypothesis testing problem as well as the form of the optimal
measurement which should be used for the certification. Although these results
may seem quite technical, they will lay the groundwork for studying the
certification of unitary channels and von Neumann measurements in further
sections.
\subsection{Certification scheme.}
Assume we are given one of two known quantum states either $\ket{\psi}$ or
$\ket{\varphi}$. The hypothesis $H_0$ corresponds to the state $\ket{\psi}$,
while the alternative hypothesis $H_1$ corresponds to the state $\ket{\varphi}$.
In other words, our goal is to decide whether the given state was $\ket{\psi}$
or $\ket{\varphi}$. To make a decision, we need to measure the given state and
we are allowed to use any POVM. We will use a quantum measurement with effects
$\{\Omega , \1- \Omega \}$, where the first effect $\Omega$ accepts the
hypothesis $H_0$ and the second effect $\1-\Omega$ accepts $H_1$.  Hence, the
probability of obtaining the type I error is given by
\begin{equation}\label{eq:pI_states}
p_\text{I}(\Omega) = \bra{\psi} (\1-\Omega) \ket{\psi}.
\end{equation}
The probability of obtaining the type II error to be minimized yields
\begin{equation}\label{eq:pII_states}
p_\text{II}=\min_{\Omega: p_{\text{I}}(\Omega) \le \delta}
\bra{\varphi} \Omega \ket{\varphi} \eqqcolon \bra{\varphi} \Omega_0 \ket{\varphi},
\end{equation}
where the minimization is performed by finding the optimal measurement
$\Omega_0$.

This problem was explored in \cite{helstrom1976quantum}. However, to keep this
work self-consistent we present in Appendix~\ref{app:states} an alternative version of the
proof.

\begin{theorem}\label{thm_pure_state}
Consider the problem of two-point certification of pure quantum states with
hypotheses given by
\begin{equation}\label{eq:hytheses_states}
\begin{split}
&H_0: \ \ket{\psi}, \\
&H_1: \ \ket{\varphi}.
\end{split}
\end{equation}
and statistical significance $\delta \in [0,1]$. Then, for the most powerful test,
the probability of the type II error \eqref{eq:pII_states} yields
\begin{equation}
p_\text{II}  = \left\{ \begin{array}{ll}
0
& \text{if} \,\,\,|\braket{\psi}{\varphi}| \leq \sqrt{\delta}, \\
\left(|\braket{\psi}{\varphi}| \sqrt{1-\delta} -
\sqrt{1-|\braket{\psi}{\varphi}|^2}
\sqrt{\delta}\right)^2
& \text{if} \,\,\, |\braket{\psi}{\varphi}| > \sqrt{\delta}.
\end{array} \right.
\end{equation}
\end{theorem}

The proof of the above theorem is presented in Appendix~\ref{app:states}. This proof gives a
construction of the optimal measurement which minimizes the probability of the
type II error. The exact form of such an optimal measurement is stated as the
following corollary.

\begin{corollary}\label{cor:_pure_state}
	The optimal strategy for two-point certification of pure quantum states
	$\ket{\psi}$ and $\ket{\varphi}$, with
	statistical significance $\delta$ yields
	\begin{enumerate}
		\item if $|\braket{\psi}{\varphi}| \leq \sqrt{\delta}$, then the optimal
		measurement is given by $\Omega_0= \ketbra{\omega}{\omega}$, where $ \ket{\omega} =
		\frac{\ket{\widetilde{\omega}}}{||\ket{\widetilde{\omega}}||} $, $
		\ket{\widetilde{\omega}} = \ket{\psi} - \braket{\varphi}{\psi}
		\ket{\varphi}$;
		\item if  $|\braket{\psi}{\varphi}| > \sqrt{\delta}$, then the optimal
		measurement is given by $\Omega_0 = \ketbra{\omega}{\omega}$ for $\ket{\omega} =
		\sqrt{1-\delta} \ket{\psi} - \sqrt{\delta} \ket{ \psi^\perp}$, $\ket{\psi^\perp}
		= \frac{\ket{\widetilde{\psi^\perp}}}{|| \ket{\widetilde{\psi^\perp}
		}||} $,
		where $ \ket{\widetilde{\psi^\perp}} = \ket{\varphi} -
		\braket{\psi}{\varphi}
		\ket{\psi}$.
	\end{enumerate}
\end{corollary}

\section{Certification of unitary
channels}\label{sec:unitary_channels_two_points}

In this section we will be interested in certification of two unitary channels
$\Phi_{U}$ and $\Phi_{V}$ for $U,V \in \UU_d$. Without loss of generality we can
assume that one of these unitary matrices  is the identity matrix and then our
task reduces to certification of channels $\Phi_\1$ and $\Phi_U$.  In the most
general case, we are allowed to use entanglement by adding an additional system.
Hence, the null hypothesis $H_0$ yields that the unknown channel is $\Phi_\1
\otimes \1$ and the alternative hypothesis $H_1$ yields that the unknown channel
is $\Phi_U \otimes \1$.

\subsection{Certification scheme}
The idea behind the scheme of certification of unitary channels is to reduce
this problem to certification of quantum states discussed in the previous
section. We prepare some (possibly
entangled) input state  $\ket{\psi}$ and perform the unknown channel on it. The
resulting state is either $\left(\1 \otimes \1 \right)\ket{\psi}$ or
$\left(U \otimes \1 \right)\ket{\psi}$. Then, we perform the measurement
$\{\Omega , \1- \Omega \}$ and make a decision whether the given channel was
$\Phi_\1 \otimes \1$ or $\Phi_U \otimes \1$. The effect $\Omega$ corresponds to
accepting $H_0$
hypothesis while $\1-\Omega$ corresponds to the alternative hypothesis~$H_1$.

The results of minimization of the probability of the type II error over input
states $\ket{\psi}$ and measurements $\Omega$ are summarized as the following
theorem. This reasoning is based on the results from Theorem
\ref{thm_pure_state}, while a related study of this problem can be found in
\cite{lu2010optimal}.

\begin{theorem}\label{thm:unitary_channels}
Consider the problem of two-point certification of unitary channels with
hypotheses
\begin{equation}\label{eq:hypotheses_unitary_channels}
\begin{split}
&H_0: \  \Phi_\1 \otimes \1, \\
&H_1: \  \Phi_U \otimes \1.
\end{split}
\end{equation}
and statistical significance $\delta \in [0,1]$. Then, for the most powerful test, the
probability of the type II error yields
\begin{equation}\label{error-for-unitary-channel}
p_{\text{II}}  = \left\{ \begin{array}{ll}
0
&\text{if} \,\,\, |\bra{\psi_0}U  \ket{\psi_0}| \leq \sqrt{\delta}, \\
\left(|\bra{\psi_0}U \ket{\psi_0}| \sqrt{1-\delta} -
\sqrt{1-|\bra{\psi_0}U \ket{\psi_0}|^2}
\sqrt{\delta}\right)^2
& \text{if} \,\,\, |\bra{\psi_0}U \ket{\psi_0}|> \sqrt{\delta},
\end{array} \right.
\end{equation}
where $\ket{\psi_0} \in \arg\min_{\ket{\psi}} |\bra{\psi} U \ket{\psi}|$.
\end{theorem}

\begin{proof}
Let us first introduce the hypotheses conditioned by the input state
$\ket{\psi}$
\begin{equation}\label{eq:conditional_hypotheses}
\begin{split}
&H_0^{\ket{\psi}}: \  \ket{\psi},  \\
&H_1^{\ket{\psi}}: \  (U \otimes \1) \ket{\psi}.
\end{split}
\end{equation}
We do not make any assumptions on the dimension of the auxiliary system for the
time being. It will turn out, however, that it suffices if its dimension equals one.
The hypotheses in \eqref{eq:conditional_hypotheses} correspond to output states
after the application of the extended unitary channel on the state $\ket\psi$.
For these hypotheses we consider the statistical significance $\delta \in
[0,1]$, that is
\begin{equation}
p_\text{I}^{\ket{\psi}}(\Omega) = \tr \left( (\1-\Omega) (\Phi_\1 \otimes
\1)(\proj{\psi}) \right)\le \delta.
\end{equation}
Our goal will be to calculate the minimized probability of the type II error

\begin{equation}\label{eq:pII_unitary_channels}
p_\text{II}=\min_{\ket{\psi}}\min_{\Omega: p_{\text{I}}^{\ket{\psi}}(\Omega)
\le \delta}
\tr(\Omega (\Phi_U \otimes \1)(\ketbra{\psi}{\psi})) \eqqcolon \tr(\Omega_0
(\Phi_U \otimes \1)(\ketbra{\psi_0}{\psi_0})),
\end{equation}
where naturally, for the optimal strategy $\ket{\psi_0}$ and $ \Omega_0$ it holds
that $p_{\text{I}}^{\ket{\psi_0}}(\Omega_0) \leq \delta$.

Now we will show that the use of entanglement is unnecessary. From
Theorem~\ref{thm_pure_state} we know that the probability of the type II error,
$p_{\text{II}}$, depends on the minimization of the inner product
$\min_{\ket{\psi}} | \bra{\psi} U \otimes \1 \ket{\psi}  |$. Directly from the
definition of numerical range we can see that $ \bra{\psi} U \otimes \1
\ket{\psi}  \in W(U\otimes \1) $. From the property of numerical range given in
Eq.~\eqref{properties-q-nr-tensor-product} and using the notation introduced in
Eq.~\eqref{eq:dist_q_nr_to_zero_def} we have
\begin{equation}
\nu \left( U \otimes \1 \right) = \nu \left( U  \right).
\end{equation}
Let $\ket{\psi_0}$ be the considered optimal input state, i.e.
$\ket{\psi_0} \in \arg\min_{\ket{\psi}} |\bra{\psi} U \ket{\psi}|$.
Therefore we can reformulate our hypotheses as
\begin{equation}
\begin{split}
&H_0^{\ket{\psi_0}}: \  \ket{\psi_0},  \\
&H_1^{\ket{\psi_0}}: \  U \ket{\psi_0}.
\end{split}
\end{equation}
These hypotheses, when taking
$\ket{\varphi} \coloneqq U \ket{\psi_0}$, were the subject of interest in
Theorem~\ref{thm_pure_state}.
\end{proof}

The next corollary follows directly from the above proof.
\begin{corollary}
Entanglement is not needed for the certification of unitary channels.
\end{corollary}

The following remark states that while considering the input state to the
certification scheme, we can restrict our attention to pure states only.
\begin{remark}
Without loss of generality, we can consider only
pure input states. The minimal value of linear objective function
\begin{equation}
\rho \mapsto \tr \left( \Omega \Phi_U(\rho) \right)
\end{equation} over a convex set $\{\rho \in \DD_d: \tr \left( \Omega
\Phi_\1(\rho) \right) \ge 1- \delta \}$ is achieved on the extremal points.
\end{remark}
\subsection{Connection with $q$-numerical range}\label{ssec:q-range} There
exists a close relationship between the above results and the definition of
numerical range, which can be seen from the proof of
Theorem~\ref{thm:unitary_channels}. It the work \cite{duan2009perfect} the
authors show the connection between the discrimination of quantum channels and
$q$-numerical range. In this section we show the connection between
certification of unitary channels and $q$-numerical range. Recall the definition
of $q$-numerical range.
\begin{equation}
W_q(X) := \{  \bra{\psi} X \ket{\varphi}:  \braket{\psi}{\varphi} =q \}.
\end{equation}
Using this notion and the notation introduced in Eq.~\eqref{eq:dist_q_nr_to_zero_def} we can rewrite our results for the probability
of the type II error from Theorem \ref{thm:unitary_channels} as
\begin{equation}\label{eq:unitary_result_alternative}
p_\text{II} = \nu^2_{\sqrt{1-\delta}}\left(U \otimes \1\right)
= \nu^2_{\sqrt{1-\delta}}\left(U \right).
\end{equation}
An independent derivation of the above formula is presented in Appendix~\ref{app:q-numerical-range}.

Let $\Theta$ be the angle between two most distant eigenvalues of a unitary
matrix $U$. Then, from the above discussion we can draw a conclusion that for
any statistical significance $\delta \in (0,1]$, if $2
\arccos\left(\sqrt{\delta}\right) \leq \Theta < \pi $, then although $\Phi_U$
and $\Phi_\1$ cannot be distinguished perfectly, they can be certified with
$p_{\text{II}} = 0 $. In other words, the numerical range $W(U)$ does not
contain zero but $\sqrt{1-\delta}$-numerical range, $W_{\sqrt{1-\delta}} (U)$,
does contain zero. The situation changes when $2
\arccos\left(\sqrt{\delta}\right) > \Theta$. Then, both numerical range $W(U)$
and $\sqrt{1-\delta}$-numerical range $W_{\sqrt{1-\delta}} (U)$ do not contain
zero. This is presented in Fig.~\ref{fig:q-numerical-range}.

\begin{figure}[h!]
\includegraphics[scale=0.7]{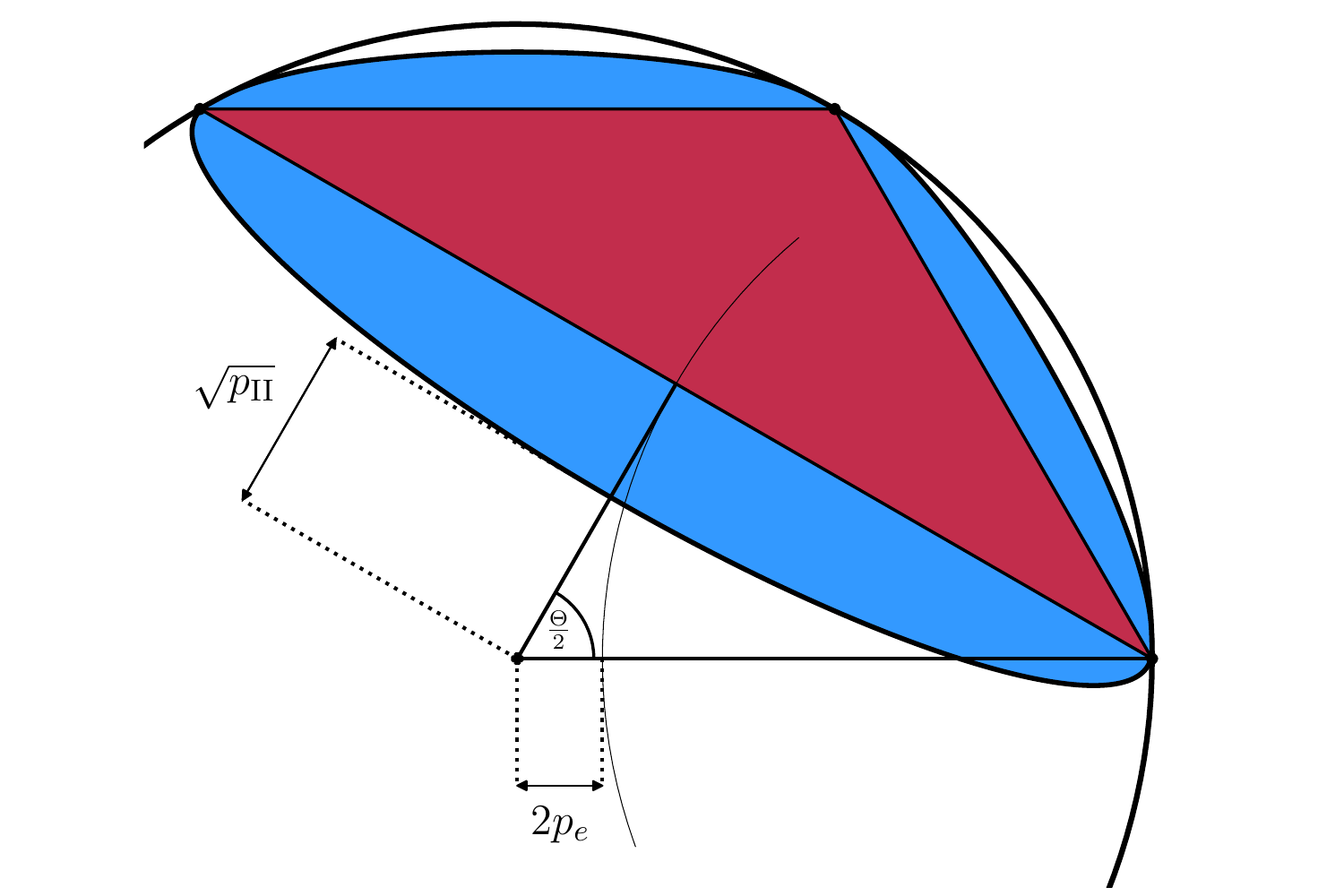}
\caption{Numerical range $W(U)$ (red triangle) and $\sqrt{1-\delta}$-numerical
range $W_{\sqrt{1-\delta}}(U)$ (blue oval)  of $U \in \UU_3 $ with
eigenvalues $1, \ee^{ \frac{\pi \ii}{3}}$ and $\ee^{ \frac{2\pi \ii}{3}}$ with
statistical significance $\delta = 0.05$. The value
$p_e$ is the probability of incorrect symmetric discrimination of
channels $\Phi_\1$ and $\Phi_U$.}
\label{fig:q-numerical-range}
\end{figure}

Now we will work towards the construction of the optimal strategy, which will be
stated as a corollary. Besides finding the optimal measurement which was shown
in previous section we will show a closed-form expression of the optimal input
state. For this purpose we will make use of the spectral decomposition of a
unitary matrix $U$ given by
\begin{equation}
U = \sum_{i=1}^d \lambda_i \ketbra{x_i}{x_i}.
\end{equation}
Let $\lambda_1, \lambda_d$ be a pair of the most distant eigenvalues of $U$. The
following corollary is analogous to the corollary from the previous section as
it presents the optimal strategy for the certification of unitary channels.

\begin{corollary}\label{cor:unitary_channels}
By $\ket{\psi_0}$ we will denote the optimal state for two-point
certification of
unitary channels and let $\ket{\varphi} \coloneqq U \ket{\psi_0}$. Then,
the
optimal strategy yields
\begin{enumerate}
\item If $0 \in W_{\sqrt{1-\delta}}(U) $, then we have two
cases
\begin{itemize}
\item  if $0 \not\in W(U)$, then we can take
\begin{equation}\label{discriminator}
\ket{\psi_0} = \frac{1}{\sqrt{2}} \ket{x_1} + \frac{1}{\sqrt{2}}
\ket{x_d}
\end{equation}
where  $\ket{x_1}$, $\ket{x_d}$ are eigenvectors corresponding to the
pair of the most distant eigenvalues $\lambda_1$, $\lambda_d$ of $U$. The
optimal measurement is given by
$\Omega_0= \ketbra{\omega}{\omega}$, where $ \ket{\omega}
= \frac{\ket{\widetilde{\omega}}}{||\ket{\widetilde{\omega}}||} $,  $
\ket{\widetilde{\omega}} = \ket{\psi_0} - \braket{\varphi}{\psi_0}
\ket{\varphi}$,
\item if $0 \in W(U)$, then we have perfect symmetric  distinguishability.
Moreover,  there
exists the probability vector $p$ such that $
\sum_{i=1}^d \lambda_i
p_i = 0 $ and we obtain that
\begin{equation}
\ket{\psi_0} = \sum_{i=1}^d \sqrt{p_i} \ket{x_i}.
\end{equation}
Analogously, we choose the optimal measurement given by  $\Omega_0=
\ketbra{\omega}{\omega}$, where $ \ket{\omega} =
\frac{\ket{\widetilde{\omega}}}{||\ket{\widetilde{\omega}}||} $, $
\ket{\widetilde{\omega}} = \ket{\psi_0} - \braket{\varphi}{\psi_0}
\ket{\varphi}$.
It easy to see that in this case we have $\Omega_0 =
\ketbra{\psi_0}{\psi_0}$.
\end{itemize}
\item
If $0 \not\in W_{\sqrt{1-\delta}}(U) $, then
the discriminator is given by Eq.~\eqref{discriminator}, whereas the optimal
measurement is can be expressed as
$\Omega_0 = \ketbra{\omega}{\omega}$ for
$\ket{\omega} = \sqrt{1-\delta} \ket{\psi_0} - \sqrt{\delta} \ket{
\psi_0^\perp}$, $\ket{\psi_0^\perp} =
\frac{\ket{\widetilde{\psi_0^\perp}}}{||
\ket{\widetilde{\psi_0^\perp} }||} $,
where $ \ket{\widetilde{\psi_0^\perp}} =
\ket{\varphi} - \braket{\psi_0}{\varphi} \ket{\psi_0}$.
\end{enumerate}
\end{corollary}

\begin{remark}
Observe that the optimal input state $\ket{\psi_0}$ does not depend on
$\delta$, while the optimal measurement $\Omega_0$ does depend on the parameter
$\delta$ in each case. It is also worth noting that the optimal state in
quantum hypothesis testing is of the same form as in the problem of unitary
channel discrimination.
\end{remark}
\section{Two-point certification of Von Neumann
measurements}\label{sec:von_neumann_measurements}

In this section we will focus on the certification of von Neumann measurements.
Recall that every quantum measurement can be associated with a
measure-and-prepare quantum channel. Therefore, while studying the certification
of quantum measurements we will often take advantage of the certification of
quantum channels discussed in the previous section, where we assumed that one
of the unitaries was the identity.
Similarly, also in the case of certification of von Neumann measurements
we will assume that one of the measurements is in
the computational basis.  Hence, we will be certifying the measurement $\PP_\1$
under the alternative hypothesis $\PP_U$.

While certifying quantum channels, the most general scenario allows for the use
of entanglement by adding an additional system. Hence, in our case of
certification of von Neumann measurements, the hypothesis $H_0$ yields that the
unknown measurement is $\PP_{\1} \otimes \1$ whereas for the alternative
hypothesis yields that the measurement is $\PP_{U} \otimes \1$.

Now we recall some technical tools which will be used to prove the main result
of this work. It was shown in~\cite[Theorem 1]{puchala2018strategies} that the
diamond norm distance between von Neumann measurements $\mathcal{P}_U$ and
$\PP_\1$ is given by

\begin{equation}\label{diamond_norm_minue} ||\PP_{U} -
\PP_\1||_\diamond = \min_{E \in \mathcal{DU}_d} ||\Phi_{UE} - 
\Phi_\1||_\diamond,
\end{equation}
where $\mathcal{DU}_d$ is the subgroup  of  diagonal unitary matrices of dimension
$d$. As we can see, the problem of discrimination of von Neumann measurements
reduces to the problem of discrimination of unitary channels. From
~\cite{watrous2018theory} we know that the diamond norm distance between two
unitary channels $\Phi_U$ and $\Phi_\1$ is expressed as
\begin{equation}\label{diamond_norm_nikwadrat}
|| \Phi_U - \Phi_\1 ||_\diamond = 2 \sqrt{1 - \nu^2\left(U \right)},
\end{equation}
where $\nu (U) = \min \{ |x|: x \in W(U) \}$.

\subsection{Certification scheme} The scenario of certification of von Neumann
measurements is as follows. We prepare some (possibly entangled) input state
$\ket{\psi}$ and, as previously, we perform the unknown von Neumann measurement
on one part of it. Then, after performing the measurement, the null hypothesis
$H_0$ corresponds to the state $\left(\PP_\1 \otimes \1 \right) (\proj{\psi})$,
while the alternative hypothesis $H_1$ corresponds to the state $\left(\PP_U
\otimes \1 \right) (\proj{\psi})$. Our goal is to find an optimal input state
and a measurement for which the probability of the type II error is saturated,
while the statistical significance $\delta$ is assumed. The results of
minimization are stated as the following theorem.

\begin{theorem}\label{thm_measurements}
Consider the problem of two-point certification of von Neumann measurements
with hypotheses
\begin{equation}\label{eq:hypotheses_measurements}
\begin{split}
&H_0: \  \PP_\1 \otimes \1 \\
&H_1: \  \PP_U \otimes \1.
\end{split}
\end{equation}
and statistical significance $\delta \in [0,1]$. Then, for the most powerful test, the
probability of the type II error yields
\begin{equation}
p_{\text{II}} = \max_{E \in \diaguni_d} \nu^2_{\sqrt{1-\delta}} \left(UE\right).
\end{equation}
\end{theorem}

It is worth mentioning that we do not make any assumptions on the dimension of
the auxiliary system, however its dimension is obviously upper-bounded by the
dimension of the input states. Additionally, the dimension of the auxiliary
system can be reduced to the Schmidt rank of the input state
$\ket{\psi}$~\cite[Proposition 4]{puchala2018strategies}. It is worth mentioning
here that in the certification of von Neumann measurements entanglement can
significantly improve the outcome of the protocol while in the case of unitary
channel certification it provides no benefit.

In contrast to the certification of unitary channels, the output
states $(\PP_\1 \otimes \1)(\proj{\psi})$ and $(\PP_U
\otimes \1)(\proj{\psi})$ are not necessarily pure. Hence the proof of the
Theorem \ref{thm_measurements} requires more advanced techniques. Luckily, we
still can make use of
the calculations from Section~\ref{sec:unitary_channels_two_points}, due to the
fact that formally mixed states
$(\PP_\1 \otimes \1)(\proj{\psi})$ and
$(\PP_U \otimes \1)(\proj{\psi})$,
conditioned by obtaining the label $i \in \{1,\ldots,d\}$, are pure.

\begin{proof}[Proof of Theorem~\ref{thm_measurements}]
In the scheme of certification of von Neumann measurements the optimized
probability of type II error can be expressed as
\begin{equation}
p_{\text{II}} \coloneqq
\min_{\ket{\psi}} \min_{\Omega: p_{\text{I}}^{\ket{\psi}}(\Omega) \le \delta}
\tr
\left( \Omega \left( \PP_U \otimes \1 \right) (\ketbra{\psi}{\psi}) \right).
\end{equation}
Our goal is to prove that
\begin{equation}\label{eq:to_prove_equality}
p_{\text{II}} = \max_{E \in \diaguni_d} \nu^2_{\sqrt{1-\delta}}
\left(UE\right).
\end{equation}

The proof is divided into two parts. In the first part we will utilize data
processing inequality presented in Lemma~\ref{data-process-in} in Appendix~\ref{app:th}. Thanks to that, we
will show the lower bound for $p_{\text{II}}$. In the second part we will use
some technical lemmas presented in Appendix~\ref{app:th} and we will utilize the results
from ~\cite{puchala2018strategies} to show the upper bound for $p_{\text{II}}$.

\hspace{1cm}
\subsubsection*{The lower bound} This part of the proof mostly will be based on
data processing inequality.
To show that
\begin{equation}
p_{\text{II}} \geq \max_{E \in \diaguni_d} \nu^2_{\sqrt{1-\delta}}
\left(UE\right)
\end{equation}
let us begin with an observation that every quantum von Neumann
measurement $\PP_U$ can be rewritten as $\Delta \circ \Phi_{(UE)^\dagger}$,
where $\Delta$ denotes the completely dephasing channel and $E \in \diaguni_d$.
Therefore, utilizing
data processing inequality in Lemma~\ref{data-process-in} in Appendix~\ref{app:th}, along with the
certification scheme of unitary channels in Theorem~\ref{thm:unitary_channels},
the optimized  probability of the type II error is lower-bounded by
\begin{equation}
p_{\text{II}}\geq \min_{\ket{\psi}} \min_{\Omega:
p_{\text{I}}^{\ket{\psi}}(\Omega) \le \delta }
\tr(\Omega (\Phi_{(UE)^\dagger}\otimes
\1)(\proj{\psi}) ) =\nu^2_{\sqrt{1-\delta}} \left( (UE)^\dagger \right) 
=\nu^2_{\sqrt{1-\delta}} \left( UE \right)
\end{equation}
which holds for each $E \in \diaguni_d$.
Hence, maximizing the value of
$\nu^2_{\sqrt{1-\delta}} \left( UE \right)$
over $E \in \diaguni_d$  leads to the lower bound of the form
\begin{equation}
p_{\text{II}}\geq\max_{E \in \diaguni_d}
\nu^2_{\sqrt{1-\delta}} \left( UE \right).
\end{equation}

\hspace{1cm}
\subsubsection*{ The upper bound}

Now we proceed to proving the upper bound.
The proof of the inequality
\begin{equation}
p_{\text{II}} \leq \max_{E \in \diaguni_d} \nu^2_{\sqrt{1-\delta}}
\left(UE\right)
\end{equation}
 will be divided into two cases depending on
diamond norm distance between considered measurements $\PP_U$
and $\PP_\1$.
In either case we will construct a strategy, that is choose a state
$\ket{\psi_0}$ and a measurement $\Omega_0$. As for every choice of $\ket{\psi}
$ and $\Omega$ it holds that
\begin{equation}
p_{\text{II}} \le \tr\left(\Omega (\PP_U \otimes \1)(\proj{\psi})\right),
\end{equation}
we will show that for some fixed $\ket{\psi_0} $ and $\Omega_0$ it holds that
\begin{equation}\label{eq:to_be_proved}
\tr\left(\Omega_0 (\PP_U \otimes \1)(\proj{\psi_0})\right) = \max_{E \in
\diaguni_d}
\nu^2_{\sqrt{1-\delta}} \left( UE \right).
\end{equation}

First we focus on the case when  $\|  \PP_U - \PP_\1
\|_\diamond = 2$. We take a state $\ket{\psi_0}$ for which it holds that
\begin{equation}\label{eq:proof_when_diamond_norm_equals_2}
\|  \PP_U - \PP_\1
\|_\diamond  = \| \left(\left(  \PP_U - \PP_\1 \right)\otimes \1 \right)
(\proj{\psi_0})
\|_1.
\end{equation}
Then, the output states $(\PP_U
\otimes \1)(\proj{\psi_0})$ and $(\PP_\1\otimes \1)(\proj{\psi_0})$ are
orthogonal and by taking the measurement $\Omega_0$ as the projection onto the
support of  $(\PP_\1\otimes \1)(\proj{\psi_0})$ we obtain
\begin{equation}\label{eq:meeee}
\tr\left(\Omega_0 (\PP_U \otimes \1)(\proj{\psi_0})\right) = 0.
\end{equation}
Let us recall that

\begin{equation}\label{diamond_norm_nikwadrat}
|| \Phi_U - \Phi_\1 ||_\diamond = 2 \sqrt{1 - \nu^2\left(U \right)},
\end{equation}
where $\nu (U) = \min \{ |x|: x \in W(U) \}$ and it holds that
\cite{puchala2018strategies}
\begin{equation}\label{diamond_norm_minue} ||\PP_{U} -
\PP_\1||_\diamond = \min_{E \in \mathcal{DU}_d} ||\Phi_{UE} - \Phi_1||_\diamond,
\end{equation}
where $\mathcal{DU}_d$ is the subgroup  of  diagonal unitary matrices of
dimension $d$.

Then, utilizing Eq.~\eqref{diamond_norm_nikwadrat} and
\eqref{diamond_norm_minue}
we obtain that
$\max_{E \in \diaguni_d} \nu^2\left( UE \right)=0$.
Therefore, by the property that $0 \in W_{\sqrt{1-\delta}}(UE)$ whenever $0 \in
W(UE)$ (see Appendix~\ref{app:q-numerical-range}), we have that
\begin{equation}
\max_{E \in \diaguni_d} \nu^2_{\sqrt{1-\delta}} \left( UE \right) =0.
\end{equation}
Secondly, we consider the situation when $\|  \PP_U - \PP_\1 \|_\diamond < 2$.
\begin{equation}\label{eq:vn-0}
E_0 \in \underset{E \in\diaguni_d}{\arg\max} \  \nu \left( UE \right).
\end{equation}
Again, by referring to Eq.~\eqref{diamond_norm_minue} and
\eqref{diamond_norm_nikwadrat} we obtain that $\nu\left( UE_0 \right)>0$.
Let $\lambda_1, \lambda_d$ be a pair of the most distant eigenvalues of
$UE_0$. Note that the following relation holds
\begin{equation}
\nu\left( UE_0 \right)
=\frac{|\lambda_1+\lambda_d|}{2}.
\end{equation}

As the assumptions of the Lemma~\ref{existence-discriminator} in Appendix~\ref{app:th} are saturated for
the defined $E_0$, we consider the input state
\begin{equation}\label{eq:measurement}
\ket{\psi_0}
=\sum_{i=1}^d \sqrt{\rho_0} \ket{i} \otimes \ket{i}
\end{equation}
where the existence of $\rho_0$ together with its properties are described in
Lemma~\ref{existence-discriminator} and Corollary~\ref{exist-discriminator} in Appendix~\ref{lem:properties-of-discriminator}.
Let us define sets
\begin{equation}
\mathcal{C}_i \coloneqq \left\{
\Omega: 0 \le \Omega \le \1, \, \tr \left(
\left(\1- \Omega \right)
\frac{\sqrt{\rho_0}\proj{i} \sqrt{\rho_0}}{\bra{i} \rho_0 \ket{i}}\right) \leq
\delta \right\}
\end{equation}
for each $i$ such that $ \bra{i}\rho\ket{i} \not= 0$.
Now we take the measurement $\Omega_0$ as
\begin{equation}\label{eq:proof_measurement_diamond_norm_less_2}
\Omega_0=
\sum_{i=1}^d \proj{i} \otimes \Omega_{i}^\top
\end{equation}
where  $\Omega_{i}\in \mathcal{C}_i$ is defined as
\begin{equation}\label{eq:vn-3}
\Omega_i \in \arg\min_{\widetilde{\Omega} \in \mathcal{C}_i }
\tr\left(\widetilde{\Omega}
\frac{\sqrt{\rho_0} U \proj{i} U^\dagger \sqrt{\rho_0}}{\bra{i}
\rho_0 \ket{i}}\right)
\end{equation}
for each $ i \in \{ 1,\ldots,d\}$ such that $\bra{i}\rho_0\ket{i} \neq 0$ and
$\Omega_i = 0$ otherwise.

Now we check that the statistical significance is satisfied, that is for the
described strategy we have
\begin{equation}
p_{\text{I}}^{\ket{\psi_0}}(\Omega_0)=1-\tr\left(\Omega_0 (\PP_\1 \otimes
\1)(\proj{\psi_0})\right)=1-\sum_{i=1}^d \tr\left(\Omega_{i} \sqrt{\rho_0}
\proj{i}
\sqrt{\rho_0}\right)\leq \delta.
\end{equation}
Hence, it remains to show that for this setting
\begin{equation}
\tr\left(\Omega_0 (\PP_U \otimes \1)(\proj{\psi_0})\right) =
\max_{E \in \diaguni_d}
\nu^2_{\sqrt{1-\delta}}\left( UE \right).
\end{equation}
Direct calculations reveal that
\begin{equation}
\begin{split}
&\tr\left(\Omega_0 (\PP_U \otimes
\1)(\proj{\psi_0})\right)=\sum_{i=1}^d \tr\left(\Omega_{i} \sqrt{\rho_0}
U\proj{i}U^\dagger
\sqrt{\rho_0}\right)\\
&=\sum_{i=1}^d \bra{i} \rho_0 \ket{i} \tr\left(\Omega_{i} \frac{\sqrt{\rho_0}
U \proj{i} U^\dagger \sqrt{\rho_0}}{\bra{i} \rho_0 \ket{i}}\right).
\end{split}
\end{equation}
Let us define
\begin{equation}
p_{\text{II}}^{|i}  =  \tr\left(\Omega_{i} \frac{\sqrt{\rho_0}
U \proj{i} U^\dagger \sqrt{\rho_0}}{\bra{i} \rho_0 \ket{i}}\right).
\end{equation}

Note that due to Corollary~\ref{lem:properties-of-discriminator} in Appendix~\ref{app:th} the absolute value of the inner 
product between pure 
states $ \frac{\sqrt{\rho_0}\ket{i}}{\| \sqrt{\rho_0}\ket{i} \|}$ and
$ \frac{\sqrt{\rho_0} U \ket{i}}{\| \sqrt{\rho_0}  \ket{i} \|}$
is the same for every $i\in \{1,\ldots,d\}: \bra{i} \rho \ket{i} \not= 0$.
Therefore we can consider the certification of pure states conditioned on the
obtained label $i$ with statistical significance $\delta$.
From the Theorem~\ref{thm_pure_state} we know that $p_\text{II}^{|i}$ depends
only on such an inner product between the certified states, hence
$p_\text{II}^{|i}=p_\text{II}^{|j}$ for each $i,j: \bra{i}\rho \ket{i}, \bra{j}
\rho \ket{j} \not= 0$.
Therefore, we have that
the value of $p_\text{II}^{|i}$ will depend on
$\left| \frac{\lambda_1 + \lambda_d }{2} \right|$.
Thus w.l.o.g.  we can assume that $p_{\text{II}}^{|1}\neq 0$ and hence
\begin{equation}
\sum_{i=1}^d \bra{i}
\rho_0 \ket{i} p_{\text{II}}^{|i}
=p_{\text{II}}^{|1}
=\tr\left(\Omega_{1} \frac{\sqrt{\rho_0}
U \proj{1} U^\dagger \sqrt{\rho_0}}{\bra{1} \rho_0 \ket{1}}\right)
\end{equation}
and in the remaining of the proof we will show that
\begin{equation}
p_{\text{II}}^{|1}= \max_{E \in \diaguni_d}
\nu^2_{\sqrt{1-\delta}}\left( UE \right).
\end{equation}
It is sufficient to study two cases depending on the relation
between
$\sqrt{\delta}$ and the inner product
\begin{equation}
\left| \frac{\bra{1} \rho_0 U \ket{1} }{\bra{1}\rho_0\ket{1}}
\right|=\left|
\frac{\lambda_1
+ \lambda_d }{2} \right|.
\end{equation}

In the case when
$\left| \frac{\lambda_1 + \lambda_d }{2} \right| \leq
\sqrt{\delta}$, then due to Theorem~\ref{thm_pure_state} we get
$p_{\text{II}}^{|1}=0$.
On the other hand, we know that
$0 \in W_{\sqrt{1-\delta}}(UE_0)$ and hence also
\begin{equation}
\max_{E \in \diaguni_d} \nu^2_{\sqrt{1-\delta}} \left( UE \right)=0.
\end{equation}

In the case when $\left| \frac{\lambda_1 + \lambda_d }{2} \right| >
\sqrt{\delta}$, then from Theorem~\ref{thm_pure_state} we
know that
\begin{equation}
p_{\text{II}}^{|1} =  \left( \left| \frac{\lambda_1 + \lambda_d
}{2} \right|\sqrt{1-\delta} -
\sqrt{1-\left| \frac{\lambda_1 + \lambda_d }{2} \right|^2}
\sqrt{\delta}\right)^2.
\end{equation}
On the other hand, for $E_0 \in \diaguni_d$ satisfying
Eq.~\eqref{eq:vn-0} we have
\begin{equation}
\nu^2_{\sqrt{1-\delta}} \left( UE_0 \right) =
\left( \left| \frac{\lambda_1 + \lambda_d
}{2} \right|\sqrt{1-\delta} -
\sqrt{1-\left| \frac{\lambda_1 + \lambda_d }{2} \right|^2}
\sqrt{\delta}\right)^2.
\end{equation}
By the particular choice of $E_0 \in \diaguni_d$, this
value is equal to
$\max_{E \in \diaguni_d}  \nu^2_{\sqrt{1-\delta}} \left( UE \right)$, hence
combining the above equations we finally obtain
\begin{equation}
p_{\text{II}}^{|1}
= \max_{E \in \diaguni_d}  \nu^2_{\sqrt{1-\delta}} \left( UE \right).
\end{equation}
To sum up, we indicated strategies $\Omega_0$ and $ \ket{\psi_0}$ for
which the optimized probability of type II error was equal to $\max_{E \in
\diaguni_d}  \nu^2_{\sqrt{1-\delta}} \left( UE \right)$. Combining this with
the previously proven inequality
\begin{equation}
p_{\text{II}}\ge \max_{E \in \diaguni_d}
\nu^2_{\sqrt{1-\delta}} \left( UE \right)
\end{equation}
gives us Eq.~\eqref{eq:to_prove_equality} and proves that the proposed strategy
$\ket{\psi_0}, \Omega_0$ is optimal.
\end{proof}

\begin{remark}
Similarly to the case of unitary channel certification, the optimal input
state $\ket{\psi_0}$ does not depend on $\delta$, while the optimal measurement
$\Omega_0$ does depend on $\delta$. Moreover, the optimal state has the same
form as in the problem of discrimination of von Neumann measurements.
\end{remark}

Finally, in Algorithm~\ref{algo:cert} present a protocol which describes the
optimal certification strategy based on the proof of
Theorem~\ref{thm_measurements}.

\begin{algorithm}[H]
\SetAlgoLined
\KwIn{Measurement $\PP$, which is either $\PP_U$ or $\PP_\Id$ and statistical 
significance $\delta$ }
\KwOut{Decision: ``Accept $H_0$'' or ``Reject $H_0$''}
  \BlankLine\BlankLine

\nl Initialize input state $\proj{\psi_0} \in \DD_{d^2}$:\\
\eIf{$||\PP_{U} - \PP_\1||_\diamond==2$}{
   $\ket{\psi_0} \coloneqq \underset{\ket{\psi} \in \C^{d^2}: 
   \braket{\psi}{\psi} = 1}{\mathtt{argmax}} 
   \|
   \left(\left(
   \PP_U - \PP_\1
   \right)\otimes \1
   \right)
(\proj{\psi}) \|_1 $\;
   }{
   $\ket{\psi_0} \coloneqq \sum_{i=1}^d \sqrt{\rho_0} \ket{i} 
   \otimes
   \ket{i}$ for $\rho_0$ defined in Corollary~\ref{lem:properties-of-discriminator} in Appendix~\ref{app:th}\;
 }

\nl Perform $\PP$ on the first subsystem of $\proj{\psi_0}$, that is $(\PP 
\otimes \1) (\proj{\psi_0})$\;
\nl Read the measurement label $i \in \{1,\ldots,d\}$\;
\nl Define resulting quantum state $\proj{\psi_i}$ conditioned by label $i$:\\
$\proj{\psi_i} \propto
\left(\bra{i} \otimes \1 \right) \left(\PP 
\otimes \1 \right) (\proj{\psi_0}) \left(\ket{i} \otimes \1 \right) $\;
\nl Prepare the measurement $\{\Gamma_i, \1 -\Gamma_i \}$ conditioned on label 
$i$: \\
\eIf{$||\PP_{U} - \PP_\1||_\diamond==2$}{
   $\Gamma_{i}$ is the projection onto the support of $\left(\bra{i} \otimes \1
   \right) \proj{\psi_0} \left(\ket{i} \otimes \1\right)$\;
   }{$\Gamma_i = \Omega_{i}^\top$ where $\Omega_i$ fullfills 
   Eq.~\eqref{eq:vn-3}\;
 }
\nl Perform the measurement $ \{ \Gamma_i, \1 - \Gamma_i \}$ on 
$\proj{\psi_i}$\;
\nl Read the measurement outcome $k \in \{1,2\}$, where label $k=1$ is 
associated 
with effect $\Gamma_i$ and the label $k=2$ with $\1 - \Gamma_i$\;
\nl \KwResult{}
\eIf{k == 1}{
   \Return ``Accept $H_0$''\;
   }{\Return ``Reject $H_0$''\;
 }
\caption{Optimal strategy for the certification of von Neumann
measurements.}\label{algo:cert}
\end{algorithm}

Below, we provide a simple example of application of Algorithm~\ref{algo:cert}
in certification of a measurement performed in the Hadamard basis.

\begin{example}
We will certify between von Neumann measurements 
$\mathcal{P}_\mathbb{H}$ and $\PP_{\Id}$, where $\mathbb{H}$ is the Hadamard 
matrix. 
\begin{enumerate}[1]
\item We calculate the distance between $\mathcal{P}_\mathbb{H}$ and 
$\PP_{\Id}$. Using semidefinite 
programming~\cite{watrous2012simpler,puchala2018strategies}
we obtain $|| \mathcal{P}_{\mathbb{H}} - 
\PP_\Id ||_\diamond = \sqrt{2}$. 
Observe that matrix $E_0$ minimizing \eqref{diamond_norm_minue} is of the 
form $E_0 = 
\frac{1}{\sqrt{2}}  \left( 
\begin{array}{cc}
 1+i&0 \\
0&-1-i
 \end{array} \right)$, which means
\begin{equation}
||\PP_{\mathbb{H}} - \PP_\Id ||_\diamond = ||\Phi_{\mathbb{H}E_0} - 
\Phi_\Id||_\diamond.
\end{equation}
In order to construct $\rho_0$ we use Lemma~\ref{existence-discriminator} in Appendix~\ref{app:th}. 
There exist states $\rho_1, \rho_2$ of the form
 $\rho_1 
= 
 \frac{1}{2}\left( \begin{array}{cc}
 1&i \\
-i&1
 \end{array} \right)
$ and 
$\rho_2 
= \frac{1}{2}  \left( \begin{array}{cc}
 1&-i \\
i&1
 \end{array} \right) $. Thus, from Corollary~\ref{lem:properties-of-discriminator}  in Appendix~\ref{app:th}
 we have 
 \begin{equation}
 \rho_0 = \frac{1}{2}(\rho_1 + \rho_2) = \frac{1}{2} \left( \begin{array}{cc}
  1&0 \\
 0&1
  \end{array} \right). \end{equation}
Hence, the input state $\ket{\psi_0}$ has a form
\begin{equation}
\ket{\psi_0} \coloneqq \sum_{i=1}^2 \frac{1}{\sqrt{2}} \ket{i} \otimes
   \ket{i}.
\end{equation}
\item 
We perform $\PP$ on the first subsystem of $\proj{\psi_0}$, that is $(\PP 
\otimes \1) (\proj{\psi_0})$.
\item  We read the measurement label either $i= 1$ or $i = 2$.
\item We reduce our problem to certification between states 
$\ket{\psi_i} = \mathbb{H} 
\ket{i}$ 
or $\ket{\psi_i} =  \ket{i}$. 
\item According to the optimal strategy for two-point certification of 
pure quantum states (Corollary
\ref{cor:_pure_state}),  we prepare 
conditional measurements $\Gamma_i$. For a fixed statistical significance 
$\delta$ and a given label
$i$, the optimal 
measurement $\Gamma_i $  is defined as
\begin{enumerate}

\item  if $\delta \ge \frac{1}{2}$, then $\Gamma_i = \mathbb{H} \proj{i^\perp} 
\mathbb{H}$.

\item if  $\delta < \frac{1}{2}$, then  $\Gamma_i = 
\ketbra{\gamma}{\gamma}$ for 
$\ket{\gamma} =
\sqrt{1-\delta} \ket{i} - \sqrt{\delta} \ket{ i^\perp}$.
\end{enumerate}
\item We perform the measurement $ \{ \Gamma_i, \1 - \Gamma_i \}$ on 
$\proj{\psi_i}$.
\item We read the measurement outcome $k \in \{1,2\}$.
\item Finally, basing on value of $k$ we make a decision whether we accept 
(for $ k = 1$) or 
reject (for $k = 2$) the null hypothesis $H_0$. 
\end{enumerate}
\end{example}

\section{Parallel multiple-shot certification}\label{sec:multiple} In this
section we focus on the scenarios in which we have access to $N$ copies of the
certified quantum objects. The copies of a given object can be used in many
configurations. The most general strategy, in the literature referred to as the
adaptive scenario, assumes that we are allowed to perform any processing between
uses of each provided copy. The adaptive scenario can be described with the
formalism of quantum networks \cite{chiribella2009theoretical}. In this paper we
restrict our attention only to the special case of adaptive strategy, when all
copies are used in parallel. This approach, known as the parallel scenario, has a
simplified description based on the tensor product of the copies of a given
quantum object. Henceforth, in our cases, the tensor product of pure states will
be again a pure state, tensor product of unitary channels -- a unitary channel
and tensor product of von Neumann measurements -- a von Neumann measurement.
Therefore, we will be able to apply our results from previous sections.

At this point, a natural question arises: can the parallel scenario be optimal? In
the case of discrimination of unitary channels and von Neumann measurements, the
positive answer was obtained \cite{chiribella2008memory,puchala2018multiple}.
Nevertheless, in some special cases of quantum channels or quantum measurements
by using an adaptive scenario we are able to improve the probability of correct
discrimination of such objects
\cite{harrow2010adaptive,krawiec2020discrimination}. In the case of
certification of quantum objects, the optimality of parallel scenario was showed
for unitary channels \cite{lu2010optimal}. In this section, we prove the
parallel approach for certification of von Neumann measurements is also optimal,
see Theorem~\ref{th:parallel_optimal}.

Let us begin with the certification of pure states. Such certification
can be understood as certifying states $\ket{\psi}^{\otimes N}$ and
$\ket{\varphi}^{\otimes N}$. The following corollary generalizes the results
from Theorem \ref{thm_pure_state}.

\begin{corollary}\label{rem_pure_state}
In the case of certification of pure states $\ket{\psi}^{\otimes N}$ and
$\ket{\varphi}^{\otimes N}$ with statistical significance
$\delta \in [0,1]$, the minimized probability of the type II error yields
\begin{equation}
p_{\text{II}}^{(N)}  = \left\{ \begin{array}{ll} 0
& |\braket{\psi}{\varphi}|^N \leq \sqrt{\delta} \\
\left(|\braket{\psi}{\varphi}|^N \sqrt{1-\delta} -
\sqrt{1-|\braket{\psi}{\varphi}|^{2N} }
\sqrt{\delta}\right)^2
&  |\braket{\psi}{\varphi}|^N > \sqrt{\delta}
\end{array} \right.
\end{equation}
where $N$ is the number of uses of the pure state.
\end{corollary}

One can note that for a given statistical significance $\delta$, by taking $N
\ge \frac{\log\sqrt{\delta}}{\log |\braket{\psi}{\varphi}|}$ we obtain
$p_\text{II}=0$. This is not in contradiction with the statement that if one
cannot distinguish states perfectly in one step, then they cannot by
distinguished perfectly in any finite number of tries, because the error is
hidden in $p_\text{I}$. This error decays exponentially, and the optimal
exponential error rate, depending on a formulation, can be stated as the Stein
bound, the Chernoff bound, the Hoeffding bound, and the Han-Kobayashi bound, see
\cite{hayashi2009discrimination} and references therein.

Secondly, we focus on the certification of unitary channels. The scenario of
parallel certification can be seen as certifying channels $\Phi_{\1^{\otimes
N}}$ and $\Phi_{U^{\otimes N}}$. Hence, for the parallel certification of such
unitary channels we have the following corollary generalizing the results from
Theorem \ref{thm:unitary_channels}.

\begin{corollary}\label{cor:unitary-paral}
In the case of parallel certification of unitary channels $\Phi_{\1^{\otimes
N}}$ and $\Phi_{U^{\otimes N}}$ with statistical
significance $\delta \in [0,1]$, the minimized probability of the type II error
yields
\begin{equation}
p_{\text{II}}^{(N)}  =
 \nu^2_{\sqrt{1-\delta}} \left(U^{\otimes N}\right)
\end{equation}
 where $N$ is the number of uses of the unitary channel.
\end{corollary}
From the above it follows that if $0 \in W_{\sqrt{1- \delta}}(U^{\otimes N})$,
then the channels $\Phi_\1^{\otimes N}$ and $\Phi_U^{\otimes N}$  can be
certified with $p_{\text{II}} = 0 $. Let $\Theta$ be the angle between a pair
of two most distant eigenvalues of a unitary matrix $U$.
The perfect certification can be achieved by taking  $ N = \ceil{\frac{2\arccos
\sqrt{\delta}}{\Theta}} $. Observe that in the special case $\delta = 0 $, we
recover the well-known formula $ N = \ceil{\frac{ \pi}{\Theta}} $ being the
number of unitary channels required for perfect discrimination in the scheme of
symmetric distinguishability of unitary channels \cite{duan2007entanglement}.
The dependence between the number $N$ of used unitary channels and the shape of
$W_{\sqrt{1-\delta}}(U^{\otimes N})$ is presented in
Fig.~\ref{fig:q-numerical-range-zero}.

\begin{figure}[h!]
\includegraphics[scale=1]{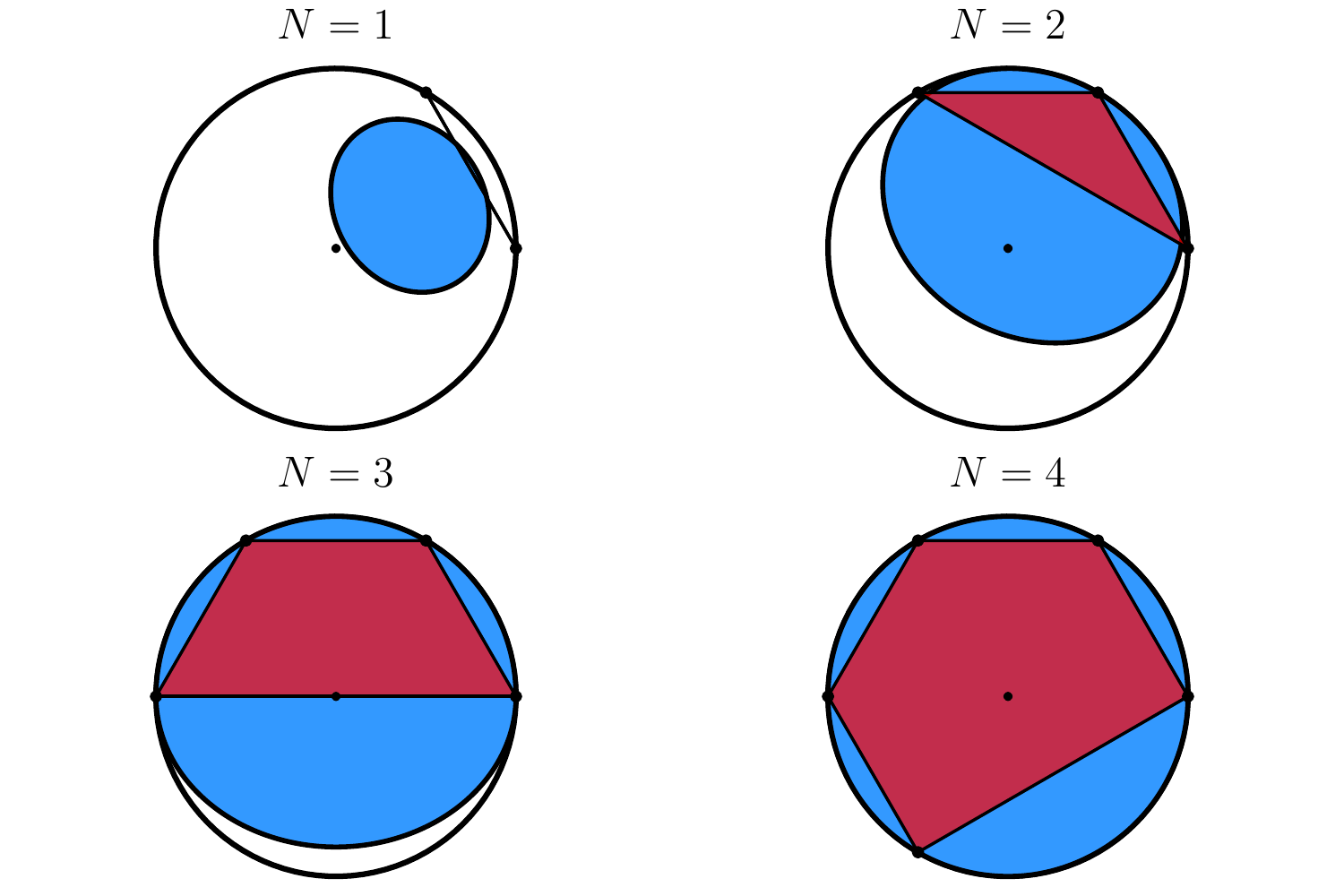}
\caption{Numerical ranges $W(U^{\otimes N})$ (polytops) and
$\sqrt{1-\delta}$-numerical
ranges $W_{\sqrt{1-\delta}}(U^{\otimes N})$ (ovals)  of $U \in
\UU_2 $ with
eigenvalues $1$ and $\ee^{ \frac{\pi \ii}{3}}$, for $N=1,2,3,4$ with
statistical significance $\delta = 0.7$.
\label{fig:q-numerical-range-zero}}
\end{figure}

We have established similar results for the case of certifying $N$ copies of
von Neumann measurements $\PP_\1$
and $\PP_U$. We consider only the parallel scenario
and therefore this can be understood as certifying von Neumann measurements
$\PP_{U^{\otimes N}}$ and $\PP_{\1^{\otimes N}}$. This issue is  studied in
Theorem~\ref{thm:parallelmeasurement}.
Moreover, it will turn out in Theorem~\ref{th:parallel_optimal} that the 
parallel scenario is optimal for the certification of von Neumann measurements.

\begin{theorem}\label{thm:parallelmeasurement}
In the case of certification of von Neumann measurements $\PP_{U^{\otimes N}}$
and $\PP_{\1^{\otimes N}}$ with statistical
significance $\delta \in [0,1]$, the minimized probability of the type II error
yields
\begin{equation}\label{eq-wqq}
p_{\text{II}}^{(N)}  = \max_{E \in \diaguni_{d}}
 \nu^2_{\sqrt{1-\delta}} \left(U^{\otimes N}E^{\otimes N}\right),
\end{equation}
where $N$ is the number of uses of the von Neumann measurements.
\end{theorem}
\begin{proof}
The von Neumann measurements $\PP_{U^{\otimes N}}$
and $\PP_{\1^{\otimes N}}$ satisfy assumptions of
Theorem~\ref{thm_measurements},  therefore we have
\begin{equation}
p_{\text{II}}^{(N)}  = \max_{E \in \diaguni_{d^N}}
 \nu^2_{\sqrt{1-\delta}} \left(U^{\otimes N}E\right).
\end{equation}
Whereas, the equality 
\begin{equation}
\max_{E \in \diaguni_{d^N}} \nu^2_{\sqrt{1-\delta}} \left(U^{\otimes N}E\right)
= \max_{E \in \diaguni_{d}}
 \nu^2_{\sqrt{1-\delta}} \left(U^{\otimes N}E^{\otimes N}\right)
\end{equation}
follows from \cite[Theorem 1]{puchala2018multiple}.
\end{proof}

Finally, we state the theorem providing the optimality of parallel scenario for
the certification of von Neumann measurements.

\begin{theorem}\label{th:parallel_optimal} The parallel scenario  for
certification of von Neumann measurements is optimal. More formally, for any
adaptive certification scenario, the probability of the type II error cannot be
smaller than in the parallel scenario.
\end{theorem}

\begin{proof}
Let us denote by $\Xi(\cdot)$ an adaptive scenario, whose input are $N$ copies
of a given quantum operation and outputs a quantum state. Let $\tilde
p_{\text{II}}^{(N)}$ be the minimized probability of the type II error for
states $\Xi\left(\PP_\1^{(N)}\right), \Xi\left(\PP_U^{(N)}\right)$. Define $d_1
\le d_2 \le \ldots \le d_N$ to be a non-decreasing sequence of natural numbers
and assume that $d_N = d_N' d_N''$ for $d_N', d_N'' \in \N$. The numbers
$d_1,\ldots,d_N$ will denote size of auxiliary systems occurring in a
construction of $\Xi\left(\PP_\1^{(N)}\right), \Xi\left(\PP_U^{(N)}\right)$. We
assume that the last auxiliary system having dimension $d_N$ is a tensor product
of two subsystems: one with dimension $d_N'$ and second with dimension $d_N''$,
which will be traced out. The general adaptive scenario for certification of $N$
copies of von Neumann measurements $\PP_U$, $\PP_\Id$ can be represented as
\cite{puchala2018multiple}:
\begin{equation}\label{eq:scenario}
\begin{split}
\Xi\left(\PP_U^{(N)}\right) = &(\1_{d^{N-1}} \otimes \PP_U \otimes \1_{d_N'} 
\otimes \tr_{d_N''}) \circ \Xi_{N-1} 
\circ \ldots
\circ(\1_d \otimes \PP_U \otimes \1_{d^{N-2}d_2}) \circ \\
&\Xi_1 \circ (\PP_U 
\otimes 
\1_{d^{N-1}d_1})\left(\proj{\psi_0}\right), \\
\Xi\left(\PP_\Id^{(N)}\right) =  &(\1_{d^{N-1}} \otimes \PP_\1 \otimes 
\1_{d_N'} 
\otimes \tr_{d_N''}) \circ \Xi_{N-1} 
\circ \ldots
\circ(\1_d \otimes \PP_\1 \otimes \1_{d^{N-2}d_2}) \circ \\
&\Xi_1 \circ (\PP_\1 
\otimes 
\1_{d^{N-1}d_1})\left(\proj{\psi_0}\right).
\end{split}
\end{equation}
The channels $\Xi_i$ are given by $\Xi_i(X) = W_i X W_i^\dagger$, 
such that
\begin{equation}
W_i = \sum_{k_1,k_2, \ldots, k_i} \proj{k_1, k_2,\ldots, k_i} \otimes V_{k_1, 
k_2, \ldots ,k_i},
\end{equation} 
where $V_{k_1, k_2, \ldots ,k_i} \in M_{d^{N-i} d_i, d^{N-i} 
d_{i+1}}$ are isometry matrices for each $i \in \{1,\ldots,N-1\}$.
The above scenario is presented in Fig.~\ref{fig:adaptive}.

\begin{figure}[h!]
\includegraphics{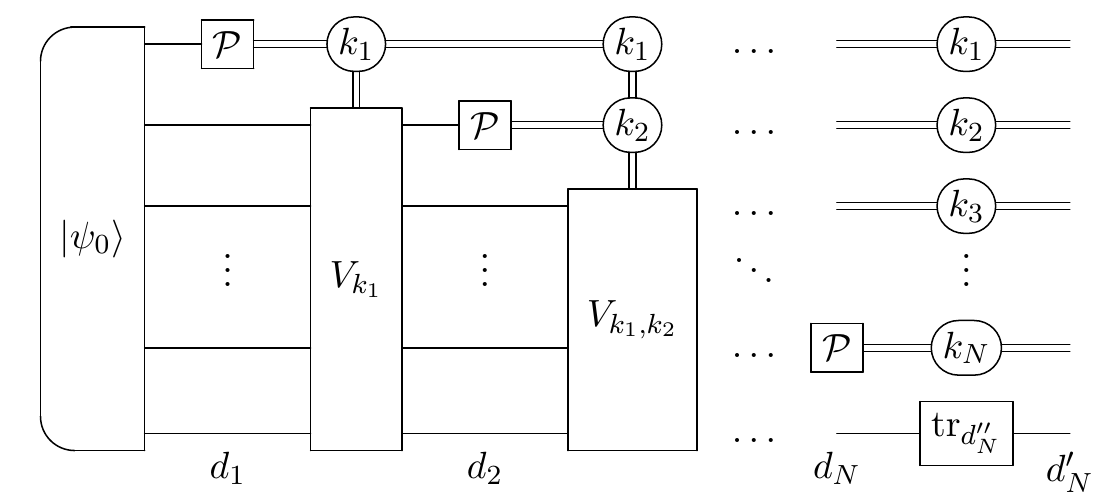}
\caption{Application of an adaptive scenario $\Xi\left(\cdot\right)$ on $N$
copies of a von Neumann measurement $\PP$ where $\PP \in \{\PP_\1, \;\PP_U\}$.
\label{fig:adaptive}}
\end{figure}

Let us decompose a von Neumann measurement $\PP_{U}$ as a composition of a
unitary channel $\Phi_{(UE_0)^\dagger} $ and the completely dephasing channel
$\Delta$, that is $\PP_U = \Delta \circ \Phi_{(UE_0)^\dagger}$, where  $E_0 $ is
a matrix maximizing Eq.~\eqref{eq-wqq}. Therefore, one can rewrite the adaptive
scenarios $\Xi\left(\PP_{U}^{(N)}\right)$ and $\Xi\left(\PP_\1^{(N)}\right)$ as
(see also Fig.~\ref{fig:adaptive2})
\begin{equation}
\begin{split}
\Xi\left(\PP_{U}^{(N)}\right) &= \left(\Delta^{\otimes N} \otimes \1_{d_N'} 
\otimes \tr_{d_N''}\right)
\circ
\Xi
\left(\Phi_{(UE_0)^\dagger}^{(N)}\right),\\
\Xi\left(\PP_{\1}^{(N)}\right) &= \left(\Delta^{\otimes N} \otimes \1_{d_N'} 
\otimes \tr_{d_N''}\right)
\circ \Xi
\left(\Phi_{\1}^{(N)}\right),\label{eq:adaptive2}
\end{split}
\end{equation} 
where 
\begin{equation}
\begin{split}
\Xi
\left(\Phi_{(UE_0)^\dagger}^{(N)}\right) = &(\1_{d^{N-1}} \otimes 
\Phi_{(UE_0)^\dagger} \otimes 
\1_{d_N}) \circ \Xi_{N-1} 
\circ \ldots
\circ(\1_d \otimes \Phi_{(UE_0)^\dagger} \otimes \1_{d^{N-2}d_2}) \circ \\
&\Xi_1 \circ (\Phi_{(UE_0)^\dagger} 
\otimes 
\1_{d^{N-1}d_1})\left(\proj{\psi_0}\right), \\
\Xi
\left(\Phi_{\1}^{(N)}\right) =  &(\1_{d^{N-1}} \otimes 
\Phi_{\1} \otimes 
\1_{d_N}) \circ \Xi_{N-1} 
\circ \ldots
\circ(\1_d \otimes \Phi_{\1} \otimes \1_{d^{N-2}d_2}) \circ \\
&\Xi_1 \circ (\Phi_{\1} 
\otimes 
\1_{d^{N-1}d_1})\left(\proj{\psi_0}\right).
\end{split}
\end{equation}

\begin{figure}[h!]
\includegraphics{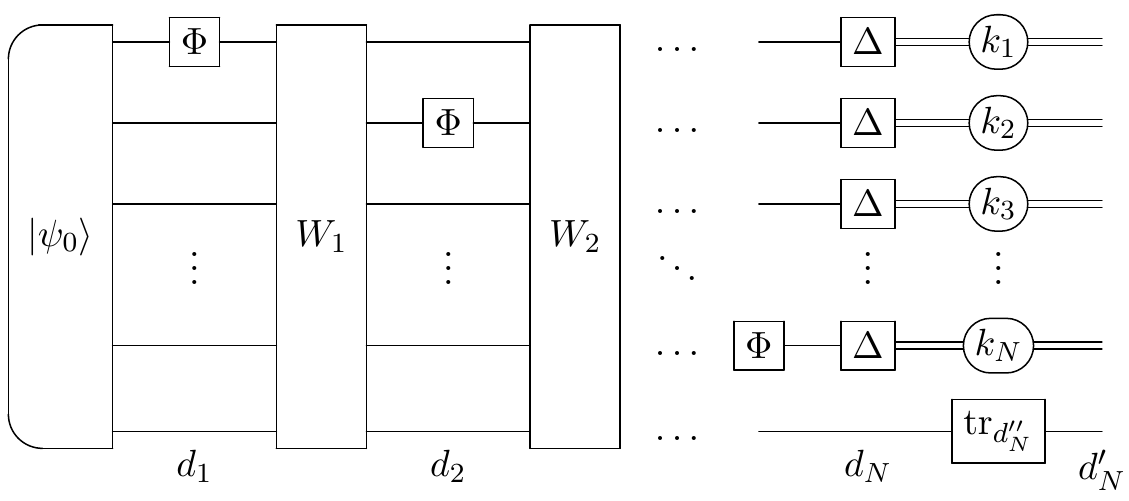}
\caption{An equivalent description of Eq.~\eqref{eq:scenario}, formalized as
Eq.~\eqref{eq:adaptive2}. The scenario $\Xi(\cdot)$ is applied to $N$ copies of
unitary channel $\Phi$, where $\Phi \in \{\Phi_\1, \Phi_{(UE_0)^\dagger}\}$. As
a last step we perform a completely dephasing channel $\Delta$.
\label{fig:adaptive2}}
\end{figure}

Let us observe that the scenarios $\Xi
\left(\Phi_{(UE_0)^\dagger}^{(N)}\right), \Xi
\left(\Phi_{\1}^{(N)}\right)$ describe certification of
$N$ copies of unitary channels $\Phi_{(UE_0)^\dagger}$ and  $\Phi_{\1}$. 
Using the data processing inequality in Lemma~\ref{data-process-in}
in Appendix~\ref{app:th}, the probability of the type II error for
certification between  $\Xi\left(\PP_{U}^{(N)}\right)$ and
$\Xi\left(\PP_{\1}^{(N)}\right)$ is no smaller that for  $\Xi
\left(\Phi_{(UE_0)^\dagger}^{(N)}\right)$ and $\Xi 
\left(\Phi_{\1}^{(N)}\right)$.

Following \cite{lu2010optimal}, the minimized probability of the type II error
for $N$-copies of unitary channels $\Phi_{(UE_0)^\dagger}$ and $\Phi_{\1}$ is 
achieved in
the parallel scenario, that is whenever $\Xi 
\left(\Phi_{(UE_0)^\dagger}^{(N)}\right)=
\Phi_{(UE_0)^\dagger}^{\otimes N} (\proj{\psi_0})$ and $\Xi
\left(\Phi_{\1}^{(N)}\right)=\Phi_\1^{\otimes N} (\proj{\psi_0})$, where 
$\ket{\psi_0}$ is an optimal state. For this 
case,
from Corollary~\ref{cor:unitary-paral} the probability of the type II error for
certification  of $N$ copies of unitary channels $\Phi_{(UE_0)^\dagger}$ and 
$\Phi_{\1}$ is
equal $\nu^2_{\sqrt{1-\delta}} \left(U^{\otimes N}E_0^{\otimes N}\right)$.
Hence, we obtain
\begin{equation}
\tilde p_{\text{II}}^{(N)}  \geq \nu^2_{\sqrt{1-\delta}} \left(U^{\otimes
N}E_0^{\otimes N}\right).
\end{equation}
From Theorem~\ref{thm:parallelmeasurement}, we have
\begin{equation}
p_{\text{II}}^{(N)}  = \nu^2_{\sqrt{1-\delta}} \left(U^{\otimes N}E_0^{\otimes
N}\right),
\end{equation}
where $p_{\text{II}}^{(N)}$ is the minimized probability of type II error for
the parallel certification of von Neumann measurements $\PP_\1$ and $\PP_U$.
Then, we have
\begin{equation}
\tilde p_{\text{II}}^{(N)} \geq p_{\text{II}}^{(N)}
\end{equation}
which finishes the proof.
\end{proof}

\section{Conclusions}\label{sec:conclusions} In this work we studied the
two-point certification of quantum states, unitary channels and von Neumann
measurements. The problem of certification of quantum objects is inextricably
related with quantum hypothesis testing. We were interested in minimizing the
probability of type II error given the upper bound on the probability of type I
error.

Although the problems of certification of quantum states and unitary channels  
are well-studied, we pointed out the connection of
certification of unitary channels with the notion of $q$-numerical range.
Afterwards, we extended this approach to the certification of
von Neumann measurements and found a formula for minimized probability of the
type II error and the optimal certification strategy. It turned out that this
formula can be also connected with the notion of $q$-numerical range.
Remarkably, it appeared
that in the case of certification of von Neumann measurements the use of
entangled input state can significantly improve the certification.

Finally, we focused on the certification of the von Neumann measurements in the
parallel scenario. More precisely, we generalized the above results for the
situation when the von Neumann measurements can be used $N$ times in parallel.
We showed that optimal certification of von Neumann measurements can be
performed without any processing additional processing, \ie\ in the parallel way.

\section*{Acknowledgments}

This work was supported by the Foundation for Polish Science (FNP) under grant
number POIR.04.04.00-00-17C1/18-00.

We would like to thank Bart{\l}omiej Gardas for fruitful discussions.

\bibliographystyle{ieeetr}
\bibliography{certification}

\appendix

\section{Certification of states}\label{app:states}
In this appendix we present the proof of Theorem~\ref{thm_pure_state}.

\begin{proof}[Proof of Theorem 1] Without loss of generality 
we can assume Eq.
$\ket{\varphi} = \alpha\ket{\psi}+\beta\ket{\psi^\perp}$, for some $\alpha,\beta
\geq0$ satisfying $\alpha^2+\beta^2=1$. For any  effect 
$\widetilde{\Omega}$ satisfying 
$\bra{\psi} \widetilde{\Omega} \ket{\psi}\ge 1 - \delta$, the
effect $\Omega$ defined as $\Omega=\Pi \widetilde{\Omega} \Pi$, where
$\Pi=\proj{\psi}+\proj{\psi^\perp}$, also satisfies the condition $\bra{\psi}
\Omega \ket{\psi}\ge 1 - \delta$ and simultaneously returns the same value of
probability of type II error. 
Hence, we can assume that rank-$2$ operator $\Omega$ satisfies 
$\Omega=\Pi\Omega\Pi$.
From the above, let $\Omega = a\Pi + b	\proj{\omega} $, where 
$\ket{\omega} = c\ket{\psi} -d\ket{\psi^\perp}$, $c \ge 0$, $d \in \C$, 
such that $c^2+ |d|^2=1$ and $a,b \in
[0,1] $, such that $a+b \le 1$. By the assumption on the value $p_{\text{I}}$,
we have 
\begin{equation}1- p_{\text{I}}(\Omega) = \bra{\psi} \Omega \ket{\psi}=
a+bc^2 \ge 1-\delta. 
\end{equation} 
Let us calculate the probability $p_{\text{II}}$:
\begin{equation}
p_{\text{II}} = \min_{\Omega: p_{\text{I}}(\Omega) \le \delta} \bra{\varphi} 
\Omega 
\ket{\varphi} = 
\min_{a,b,c,d \in \mathcal{A}}
\left( \alpha^2(a+bc^2) + \beta^2(a + 
b|d|^2) - 2 \alpha \beta bc \Re(d) \right)
\end{equation}
where $\mathcal{A} := \{ a,b,c,d:  a+b \le 1,\  a+bc^2 \ge 1-\delta, \ c^2+
|d|^2=1, \ a,b, c \in [0,1],\  d \in \C\}$.
Note that the above formula is minimized when $d\in \mathbb{R}$ is nonnegative.
Hence
\begin{equation}
\bra{\varphi} \Omega \ket{\varphi} = a+b \left( \alpha c - 
\beta d \right)^2.
\end{equation}
Thus, our task reduces to minimizing the formula 
\begin{equation}
p_\text{II} = 
\min_{a,b,c \in \mathcal{B}}
a+b \left( \alpha c - 
\beta \sqrt{1-c^2} \right)^2
\end{equation}
where 
$\mathcal{B}:= \{  a,b, c \in [0,1], \  a+b \le 1, \  a+bc^2 \ge 1-\delta \}$.
We consider two cases.
\begin{enumerate}
\item  
If $\alpha \leq \sqrt{\delta}$, then we take $a=0,b=1,c=\beta, d =
\sqrt{1-\beta^2}$. In this case $a, b, c \in \mathcal{B}$
and we obtain $p_{\text{II}}=0$. The optimal strategy is represented
by effect $ \Omega_0 = \proj{\omega}$, where $\ket{\omega} = \beta\ket{\psi} -
\alpha\ket{\psi^\perp}$.
\item 		
Let $\alpha > \sqrt{\delta}$ and take $a=0,b=1,c=\sqrt{1-\delta}, d = 
\sqrt{\delta}$. Again $a, b, c \in \mathcal{B}$ and
$p_{\text{II}}=\left(\alpha \sqrt{1-\delta} - \beta \sqrt{\delta}\right)^2$. The
optimal strategy is represented by effect $ \Omega_0 = \proj{\omega}$ where
$\ket{\omega} = \sqrt{1-\delta}\ket{\psi} - \sqrt{\delta}\ket{\psi^\perp}$. 
The optimality of this value can be checked by using standard constrained 
optimization techniques. 
\end{enumerate}
\end{proof}

\section{$q$-numerical range and certification of unitary channels}\label{app:q-numerical-range}

\subsection{$q$-numerical range in the problem of two-point certification of unitary channels}
In this appendix we will present an alternative derivation the result for the 
probability of the 
type II error in the certification of unitary channels given in Eq.~\ref{error-for-unitary-channel}.
 
We would like to bound the probability of the type I error by $\delta$, that is 
$p_{\text{I}}^{\ket{\psi}}(\Omega) = \tr ((\1 - \Omega)\proj{\psi}) \leq 
\delta$. Let us consider $\Omega = \ketbra{\omega}{\omega}$.
Hence, we have 
\begin{equation}
\tr \left(\Omega \proj{\psi} \right) = |\braket{\omega}{\psi} |^2 \geq 1 - 
\delta.
\end{equation}
The probability of the type II error takes the form
\begin{equation}
\begin{split}
p_{\text{II}} 
&= \min_{\ket{\psi}} \min_{\Omega: p_{\text{I}}^{\ket{\psi}}(\Omega) \leq 
\delta}
\tr \left(\Omega (U \otimes \1) \proj{\psi} (U^\dagger \otimes \1)\right) \\
& = \min_{\ket{\psi}} 
\min_{\ket{\omega}: p_{\text{I}}^{\ket{\psi}}(\proj{\omega}) \leq \delta}
\bra{\psi} (U^\dagger \otimes \1) \proj{\omega} (U \otimes \1) \ket{\psi} \\
&= \min_{\ket{\psi}} 
\min_{\ket{\omega}: p_{\text{I}}^{\ket{\psi}}(\proj{\omega}) \leq \delta}
|\bra{\psi} (U \otimes \1) \ket{\omega} |^2.
\end{split}
\end{equation}
Let us recall that the $q$-numerical range is defined as
\begin{equation}
W_q (A)= \{\bra{\xi_0} A \ket{\xi_1}: \braket{\xi_0}{\xi_1} = q \}
\end{equation}
and we use the notation
\begin{equation}\label{eq:dist_q_nr_to_zero_def}
\nu_{q}(X)= \min\{ |x|: x \in W_q(X) \}.
\end{equation}
Now from the definition of the $q$-numerical range for $q = \sqrt{1-\delta}$ 
and its properties~\cite{duan2009perfect}
\begin{equation}\label{properties-q-nr-inclusions}
W_{q'} \subseteq \frac{q'}{q} W_{q} \quad \text{for} \quad q \leq q', \quad 
q,q' \in \R 
\end{equation} and 
\begin{equation}\label{properties-q-nr-tensor-product}
W_q (X \otimes \1) = W_q(X), \quad q\in \R
\end{equation}
it easy to see that \begin{equation}
\nu_q (X \otimes \1) = \nu_q(X), \quad q\in \R,
\end{equation}
which will imply that
\begin{equation}
p_{\text{II}}  = 
\nu^2_{\sqrt{1-\delta}}\left( U \otimes \1 \right)
=\nu^2_{\sqrt{1-\delta}}\left( U  \right).
\end{equation}

Therefore, we conclude that the use of entanglement for the case of
certification of unitary channels does not improve the certification.

\subsection{Distance of $q$-numerical range to zero}
In this subsection we will focus on calculating the distance from the 
$q$-numerical range the to the origin of the coordinate system.
Let us begin with the two-dimensional case when the unitary matrix $U$ has two 
eigenvalues $\lambda_1$ and $\lambda_2$. Without loss of generality we can 
assume $\lambda_1 = 1$.	
From \cite{li1998some} we know that the $q$-numerical range is an 
elliptical disc with eccentricity equal to $q$ and foci $q \lambda_1$ and 
$q \lambda_2$, see Fig \ref{fig:ellipse}. Let $c$ denote the distance from 
the center of the ellipse to 
the focus and $a$ be the distance from the center of the ellipse to its vertex.
Using this notation the eccentricity yields $q = c/a$.	
Let $b$ denote the distance from the center of the ellipse to its co-vertex, 
which it the point which saturates the minimum.	

\begin{figure}[h!]
\includegraphics[scale=0.7]{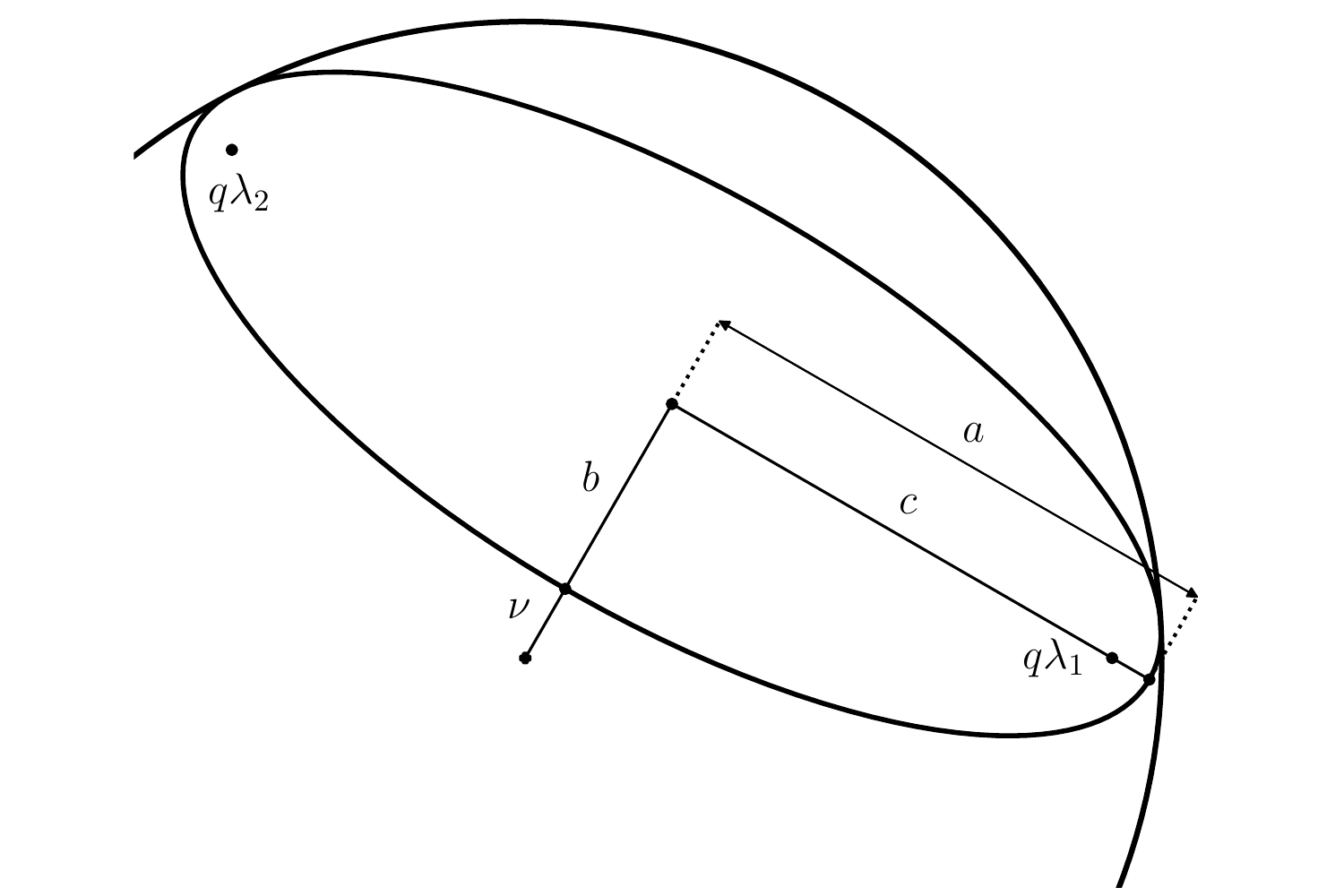}
\caption{Schematic illustration of an ellipse and notation used in Appendix, 
where we use shortcut notation $\nu \coloneqq \nu_{q}(U)$.}\label{fig:ellipse}
\end{figure}

First, we will calculate $b$. We note that
\begin{equation}
c = \frac{1}{2} \left\Vert q \lambda_1 - q \lambda_2 \right\Vert 
= \frac{q}{2} \left\Vert \lambda_1 - \lambda_2 \right\Vert 
= \frac{\sqrt{1-\delta}}{2} \left\Vert \lambda_1 - \lambda_2 \right\Vert.
\end{equation}
From the properties of the ellipse and the form of the eccentricity  $q$ we have
\begin{equation}
b = \sqrt{a^2 - c^2} = \sqrt{\frac{c^2}{q^2}-c^2} = c \sqrt{\frac{1}{q^2} - 1}
= c \sqrt{\frac{1}{1-\delta} - 1}= c \sqrt{\frac{\delta}{1-\delta}}.
\end{equation}
Hence
\begin{equation}
b = \frac{\sqrt{1-\delta}}{2} \left\Vert \lambda_1 - \lambda_2 \right\Vert
 \sqrt{\frac{\delta}{1-\delta}}
 = \frac{\sqrt{\delta}}{2} \left\Vert \lambda_1 - \lambda_2 \right\Vert.
\end{equation}
On the other hand we have
\begin{equation}
\nu_q(U) +b = \left\Vert \frac{q\lambda_1 + q\lambda_2}{2} \right\Vert
= \frac{q}{2} \left\Vert \lambda_1 + \lambda_2 \right\Vert
= \frac{\sqrt{1-\delta}}{2} \left\Vert \lambda_1 + \lambda_2 \right\Vert
\end{equation}
and therefore
\begin{equation}
\begin{split}
\nu_q(U) 
&= \frac{\sqrt{1-\delta}}{2} \left\Vert \lambda_1 + \lambda_2 
\right\Vert
- \frac{\sqrt{\delta}}{2} \left\Vert \lambda_1 - \lambda_2 \right\Vert \\
&= \frac{1}{2} \left( \sqrt{1-\delta}\left\Vert \lambda_1 + \lambda_2 
\right\Vert
- \sqrt{\delta} \left\Vert \lambda_1 - \lambda_2 \right\Vert \right).
\end{split}
\end{equation}	

Now we need to show that the above expression for the distance $\nu_q \left( U
\right)$ is valid also for higher dimensions. The boundary of $q$-numerical
ranges for larger matrices is described in \cite{li1998some}. It consists of
parts of a few ellipses obtained is an analogous way. Let $\lambda_1$ and
$\lambda_d$ be the pair of the most distant eigenvalues of $U$. Let $\lambda_i$
and $\lambda_j$ bo some pair of eigenvalues such that $i,j \not = 1,d$. Let
$\widetilde{\nu}_q \left( U \right)$ be the distance from zero the ellipse built
on $\lambda_i$ and $\lambda_j$ in the same way as above. Our goal is to prove
that $\widetilde{\nu}_q \left( U \right) > \nu_q \left( U \right)$.

We note that $\left\Vert \lambda_1 - \lambda_2 \right\Vert > \left\Vert
\lambda_i - \lambda_j \right\Vert$. Hence to prove that $\widetilde{\nu}_q
\left( U \right) > \nu_q \left( U \right)$ it suffices to show  that $\left\Vert
\lambda_1 + \lambda_2 \right\Vert < \left\Vert \lambda_i + \lambda_j
\right\Vert$. As all the eigenvalues lie on the unit circle, the from the
parallelogram law we have $\left\Vert \lambda_1 + \lambda_2 \right\Vert^2  
= 4- \left\Vert \lambda_1 - \lambda_2 \right\Vert^2$.
Therefore
\begin{equation}
\begin{split}
\left\Vert \lambda_1 + \lambda_2 \right\Vert
&= \sqrt{4- \left\Vert \lambda_1 - \lambda_2 \right\Vert^2}
<  \sqrt{4- \left\Vert \lambda_i - \lambda_j \right\Vert^2} \\
&= \sqrt{4- \left( 4- \left\Vert \lambda_i + \lambda_j \right\Vert^2 \right)}
= \left\Vert \lambda_i + \lambda_j \right\Vert.
\end{split}
\end{equation}
and thus $\widetilde{\nu}_q \left( U \right) > \nu_q \left( U \right)$, from 
which it 
follows that 
\begin{equation}
\nu_{\sqrt{1-\delta}} \left( U \right)
= \frac{1}{2} \left( \sqrt{1-\delta}\left\Vert \lambda_1 + \lambda_d \right\Vert
- \sqrt{\delta} \left\Vert \lambda_1 - \lambda_d \right\Vert \right)
\end{equation}	
holds for any dimension $d$. 
The above formula can be easily translated into trigonometric functions where 
$\Theta$ is the angle between $\lambda_1$ and $\lambda_d$. Hence, we have
\begin{equation}
\nu_{\sqrt{1-\delta}} \left( U \right)
= \sqrt{1-\delta}\cos \left(\frac{\Theta}{2}\right) - \sqrt{\delta}  \sin 
\left(\frac{\Theta}{2}\right).
\end{equation}	

Therefore,
\begin{equation}
p_\text{II} 
= \nu^2_{\sqrt{1-\delta}} \left( U \otimes \1 \right)
= \nu^2_{\sqrt{1-\delta}} \left( U \right)
= \left( \sqrt{1-\delta}\cos \left(\frac{\Theta}{2}\right) -
\sqrt{\delta}  \sin \left(\frac{\Theta}{2}\right)  \right)^2.
\end{equation}

\section{Certification of von Neumann measurements}\label{app:th}
%

In this appendix we recall a few technical lemmas  necessary to prove the main 
theorem in the paper.
The first lemma is the data processing inequality. This inequality, along with 
its proof, can be found eg. in \cite{wang2012one}. However, to keep this 
work self-consistent we present our modified version of them.
\begin{lemma}(Data processing inequality)\label{data-process-in}
Let $\delta > 0 $ and $\Omega$ be a positive semidefinite operator such that 
$\Omega \le \1$. For any quantum channel $\Phi$ and quantum states $\rho, 
\sigma$ the following 
holds
\begin{equation}
\min_{\Omega: \tr (\Omega \rho) \ge 1-\delta}  \tr ( \Omega \sigma) \le 
\min_{\Omega: \tr (\Omega \Phi(\rho)) \ge 1-\delta}  \tr ( \Omega 
\Phi(\sigma)).
\end{equation}
\end{lemma}

\begin{proof}
Let us consider two-point certification of two quantum states $\rho$ and 
$\sigma$
with statistical significance $\delta$. To calculate the probability of the 
type II error, $p_{\text{II}}$, we formulate the problem as 
\begin{equation}
\min_{\Omega:\\ \tr (\Omega \rho) \ge 1-\delta}  \tr ( \Omega \sigma).
\end{equation}

Now, consider the scenario in which we use as processing the quantum channel
$\Phi$ on states $\rho$ and $\sigma$. We want to calculate
\begin{equation}
\min_{\Omega: \tr (\Omega \Phi(\rho)) \ge 1-\delta}  \tr ( \Omega 
\Phi(\sigma))
\end{equation}
which is equivalent to
\begin{equation}
\min_{\Omega: \tr (\Phi^\dagger(\Omega)\rho) \ge 1-\delta}  \tr ( 
\Phi^\dagger(\Omega) 
\sigma).
\end{equation}
It easy to see that $\Phi^\dagger(\Omega)$ is also a measurement and 
\begin{equation}
\{ \Phi^\dagger(\Omega): \tr (\Phi^\dagger(\Omega) \rho) \ge 1-\delta \} 
\subseteq \{ 
\Omega: \tr 
(\Omega \rho) \ge 1-\delta \}.
\end{equation} 
Eventually, we obtain the data processing inequality given by 
\begin{equation}
\min_{\Omega: \tr (\Omega \rho) \ge 1-\delta}  \tr ( \Omega \sigma) \le 
\min_{\Omega: \tr (\Omega \Phi(\rho)) \ge 1-\delta}  \tr ( \Omega 
\Phi(\sigma)).
\end{equation}
\end{proof}
The following lemma is proved in the work~\cite{puchala2018strategies}.
\begin{lemma}(Lemma 5 from~\cite{puchala2018strategies}, {\it direct 
	implication})\label{existence-discriminator}
Assume that $E_0 \in \DU_d$ satisfies the condition
\begin{equation}
||\Phi_{UE_0} - 
\Phi_\1||_\diamond=||\PP_{U} - \PP_\1||_\diamond < 2. 
\end{equation}
Let $\lambda_1, \lambda_d$ be a pair of the most distant eigenvalues of 
$UE_0$
and $\Pi_1,\Pi_d$ be the projectors   onto  the   subspaces spanned  
by  the
eigenvectors corresponding  to $\lambda_1$ and $\lambda_d$, respectively. 
Then, there exist states $\rho_1, \rho_d $, satisfying the 
following conditions  
\begin{equation}\label{exist-discriminator}
\begin{split}
\rho_1 &= \Pi_1 \rho_1 \Pi_1 \\
\rho_d &= \Pi_d \rho_d \Pi_d \\
\diag(\rho_1) &= \diag(\rho_d).
\end{split}
\end{equation}
\end{lemma}
The next corollary follows directly from above lemma.
\begin{corollary}\label{lem:properties-of-discriminator}
Let $\rho_0=\frac12 \rho_1 + \frac12 \rho_d$ be the state satisfying 
conditions 
given by Eq.~\eqref{exist-discriminator}.	Then, for each $i \in 
\{1,\ldots,d\}$ 
we have
\begin{equation}\label{eq:vn-4}
\tr\left(\sqrt{\rho_0}\proj{i} \sqrt{\rho_0}\right)=\tr\left(\sqrt{\rho_0} 
U 
\proj{i} 
U^\dagger \sqrt{\rho_0}\right).
\end{equation}
Moreover, for each $i\in \{1,\ldots,d\}$ such that 
$\bra{i}\rho_0\ket{i}\neq0$ we get
\begin{equation}\label{eq:vn-5}
\left| \frac{\bra{i} \rho_0 U \ket{i} }{\bra{i}\rho_0\ket{i}} 
\right|
=\left| \frac{\lambda_1 + \lambda_d }{2} \right|.
\end{equation}
\end{corollary}

\begin{proof}
Let $U= \sum_{i=1}^d \lambda_i \Pi_i$, where $\{\Pi_i\}_{i=1}^d$ is a set 
of orthogonal projectors. Then
\begin{equation}
\begin{split}
&\tr \left( \sqrt{\rho_0} U \proj{i} U^\dagger\sqrt{\rho_0}\right)
= \bra{i} U^\dagger \rho U \ket{i}
= \bra{i} U^\dagger \left( \frac{1}{2}\rho_1 + \frac{1}{2}\rho_d\right) U 
\ket{i} \\
&=\bra{i} U^\dagger \left( \frac{1}{2} \Pi_1\rho_1 \Pi_1 + 
\frac{1}{2} \Pi_d \rho_d \Pi_d \right) U \ket{i} \\
&=\bra{i} \left(\sum_{i=1}^d \overline{\lambda_i} \Pi_i^\dagger\right) 
\left( 
\frac{1}{2} \Pi_1\rho_1 \Pi_1 + 
\frac{1}{2} \Pi_d \rho_d \Pi_d \right) \left(\sum_{i=1}^d \lambda_i 
\Pi_i\right) \ket{i} \\
&=  \bra{i} \left( \frac{1}{2}\rho_1 + \frac{1}{2}\rho_d\right)  \ket{i}
= \tr \left( \sqrt{\rho_0}  \proj{i} \sqrt{\rho_0}\right).
\end{split}
\end{equation}
where the third equality follows from Lemma \ref{existence-discriminator}.

To prove the second part of the proposition we calculate
\begin{equation}
\begin{split}
\left\vert  \frac{\bra{i} \rho_0 U \ket{i}}{\bra{i}\rho_0\ket{i}} 
\right\vert 
&= \left\vert  \frac{\bra{i} \left( \frac{1}{2}\rho_1 + 
	\frac{1}{2}\rho_d\right) \left(\sum_{i=1}^d \lambda_i \Pi_i \right)
	\ket{i}}{\bra{i}\rho_0\ket{i}} \right\vert \\
&= \left\vert  \frac{\bra{i}  \sum_{i=1}^d \lambda_i \left( \frac{1}{2} 
	\Pi_1\rho_1\Pi_1 + \frac{1}{2}   \Pi_d \rho_d \Pi_d\right) \Pi_i 
	\ket{i}}{\bra{i}\rho_0\ket{i}} \right\vert \\
&= \left\vert  \frac{\bra{i}   \left( \frac{1}{2} \lambda_1 
	\Pi_1\rho_1\Pi_1 + \frac{1}{2} \lambda_d  \Pi_d \rho_d \Pi_d\right) 
	\ket{i}}{\bra{i}\rho_0\ket{i}} \right\vert \\
&= \left\vert  \frac{\bra{i}   \left( \frac{1}{2} \lambda_1 
	\rho_1 + \frac{1}{2} \lambda_d  \rho_d \right) 
	\ket{i}}{\bra{i}\rho_0\ket{i}} \right\vert
=  \left\vert  \frac{\lambda_1 + \lambda_d}{2} \right\vert. 
\end{split}
\end{equation}
\end{proof}

\section{Animation of $q$-numerical range}\label{app:animation}
For an animation of the behavior of the $q$-numerical range of a unitary matrix $U \in \UU_3$ see the attached \verb|gif| file.

\begin{figure}[!htp]
	\caption{An animation of $q$-numerical range of unitary matrix $U \in \UU_3$
	with eigenvalues $1, \ee^{ \frac{\pi \ii}{3}}$ and $\ee^{ \frac{2\pi
	\ii}{3}}$ for all parameters $q \in [0,1]$.}
\end{figure}

\end{document}